\documentclass[letterpaper,11pt]{article}

\usepackage{lineno}

\usepackage{fontawesome}

\usepackage[longnamesfirst,numbers]{natbib}
\usepackage{url}
\usepackage{xspace}

\usepackage{diagbox}
\usepackage{makecell}
\usepackage{booktabs}
\usepackage{tikz,tikz-3dplot}
\usepackage{amsmath}
\usepackage{mathtools}
\usepackage{bbm}
\usepackage{amsfonts}
\usepackage{amssymb}
\usepackage{amsthm}
\usepackage{url,ifthen}
\usepackage[shortlabels]{enumitem}
\usepackage{srcltx}
\usepackage{dsfont}
\usepackage{multirow}
\usepackage{boxedminipage}
\usepackage[margin=1.1in]{geometry}
\usepackage{nicefrac}
\usepackage{xspace}
\usepackage{graphicx}
\usepackage{subcaption}
\usepackage{color}
\usepackage{colortbl}
\usepackage{setspace}
\usepackage{pgfplots}
\pgfplotsset{compat=1.12}
\usepackage{algorithm,algorithmic}
\usepackage[algo2e]{algorithm2e}
\RestyleAlgo{ruled}

\usepackage{multicol}

\usepackage{xcolor}
\definecolor{DarkGreen}{rgb}{0.1,0.5,0.1}
\definecolor{DarkRed}{rgb}{0.5,0.1,0.1}
\definecolor{DarkBlue}{rgb}{0.1,0.1,0.5}
\definecolor{DarkishBlue}{RGB}{89, 131, 146}
\definecolor{Gray}{rgb}{0.2,0.2,0.2}

\definecolor{c1}{RGB}{38, 70, 83}
\definecolor{c2}{RGB}{42, 157, 143}
\definecolor{c3}{RGB}{233, 196, 106}
\definecolor{c5}{RGB}{231, 111, 81}
\definecolor{c4}{RGB}{244, 162, 97}

\usepackage{listings}
\lstdefinestyle{mystyle}{
    commentstyle=\color{DarkBlue},
    keywordstyle=\color{DarkRed},
    numberstyle=\tiny\color{Gray},
    stringstyle=\color{DarkGreen},
    basicstyle=\footnotesize,
    breakatwhitespace=false,         
    breaklines=true,                 
    captionpos=b,                    
    keepspaces=true,                 
    numbers=left,                    
    numbersep=5pt,                  
    showspaces=false,                
    showstringspaces=false,
    showtabs=false,                  
    tabsize=2
}
\usepackage{subcaption}
\usepackage[pdftex]{hyperref}
\usepackage[capitalise,noabbrev]{cleveref}
\usepackage{autonum} %
\hypersetup{
    unicode=false,          %
    pdftoolbar=true,        %
    pdfmenubar=true,        %
    pdffitwindow=false,      %
    pdfnewwindow=true,      %
    colorlinks=true,       %
    linkcolor=DarkishBlue,          %
    citecolor=DarkGreen,        %
    filecolor=DarkRed,      %
    urlcolor=DarkBlue,          %
    pdftitle={},
    pdfauthor={},
}

\def\draft{1}

\def\submit{0}

\ifnum\draft=1 %
    \def\ShowAuthNotes{1}
\else
    \def\ShowAuthNotes{0}
\fi

\ifnum\submit=1
\newcommand{\forsubmit}[1]{#1}
\newcommand{\forreals}[1]{}
\else
\newcommand{\forreals}[1]{#1}
\newcommand{\forsubmit}[1]{}
\fi

\ifnum\ShowAuthNotes=1

\else
\fi

\def\isBiometrika{0}
\ifnum\isBiometrika=1
\newcommand{\proofref}[1]{of \cref{#1}}
\else
\newcommand{\proofref}[1]{Proof of \cref{#1}}
\fi

\ifnum\isBiometrika=1
\newcommand{\paragraphref}[1]{\noindent\textbf{#1}}
\else
\newcommand{\paragraphref}[1]{\paragraph{#1}}
\fi

\newtheorem{theorem}{Theorem}[section]
\newtheorem{remark}[theorem]{Remark}
\newtheorem{lemma}[theorem]{Lemma}
\newtheorem{corollary}[theorem]{Corollary}
\newtheorem{proposition}[theorem]{Proposition}

\newtheorem{assumption}[theorem]{Assumption}

\newtheorem{definition}[theorem]{Definition}%

\newtheorem*{definition*}{Definition}
\newtheorem*{proposition*}{Proposition}

\theoremstyle{definition}
\newtheorem*{example*}{Example}

\newtheoremstyle{example_contd}
{\topsep} {\topsep}%
{}%
{}%
{\bfseries}%
{.}%
{1em}%
{\thmname{#1} \thmnumber{ #2}\thmnote{#3} (continued)}%

\theoremstyle{example_contd}

\newcommand{\chapterref}[1]{\hyperref[ch:#1]{Chapter~\ref{ch:#1}}}

\newcommand{\claimref}[1]{\hyperref[claim:#1]{Claim~\ref{claim:#1}}}

\newcommand{\corollaryref}[1]{\hyperref[cor:#1]{Corollary~\ref{cor:#1}}}

\newcommand{\definitionref}[1]{\hyperref[def:#1]{Definition~\ref{def:#1}}}

\newcommand{\equationref}[1]{\hyperref[eq:#1]{Equation~\ref{eq:#1}}}

\newcommand{\factref}[1]{\hyperref[fact:#1]{Fact~\ref{fact:#1}}}

\newcommand{\figureref}[1]{\hyperref[fig:#1]{Figure~\ref{fig:#1}}}

\newcommand{\tableref}[1]{\hyperref[tab:#1]{Table~\ref{tab:#1}}}

\newcommand{\itemref}[1]{\hyperref[item:#1]{Item~(\ref{item:#1})}}

\newcommand{\lemmaref}[1]{\hyperref[lem:#1]{Lemma~\ref{lem:#1}}}

\newcommand{\propref}[1]{\hyperref[prop:#1]{Proposition~\ref{prop:#1}}}

\newcommand{\propositionref}[1]{\hyperref[prop:#1]{Proposition~\ref{prop:#1}}}

\newcommand{\remarkref}[1]{\hyperref[rem:#1]{Remark~\ref{rem:#1}}}

\newcommand{\sectionref}[1]{\hyperref[sec:#1]{Section~\ref{sec:#1}}}

\newcommand{\theoremref}[1]{\hyperref[thm:#1]{Theorem~\ref{thm:#1}}}

\newcommand{\assumptionref}[1]{\hyperref[ass:#1]{Assumption~\ref{ass:#1}}}

\newcommand{\exampleref}[1]{\hyperref[exmp:#1]{Example~\ref{exmp:#1}}}

\newcommand{\algoref}[1]{\hyperref[algo:#1]{Algorithm~\ref{algo:#1}}}

\usepackage[utf8]{inputenc}
\usepackage[T1]{fontenc}
\usepackage{kpfonts}
\usepackage{microtype}

\newcommand{\Esymb}{\mathbb{E}}

\newcommand{\E}{\Esymb}
\DeclareMathOperator*{\EE}{\mathbb{E}}
\DeclareMathOperator*{\Var}{\mathrm{Var}}

\newcommand{\Ex}[1]{\E\left[#1\right]}

\newcommand{\widebar}[1]{\overline{#1}}
\newcommand{\tv}{\text{TV}}
\newcommand{\var}[1]{\Var\Paren{#1}}

\newcommand{\flatfrac}[2]{#1/#2}

\newcommand{\mper}{\,.}
\newcommand{\mcom}{\,,}

\newcommand{\cA}{{\cal A}}

\newcommand{\cD}{{\cal D}}

\newcommand{\cP}{{\cal P}}

\newcommand{\Paren}[1]{\left(#1 \right )}

\newcommand{\Brac}[1]{\left[#1 \right]}

\newcommand{\Set}[1]{\left\{#1\right\}}

\newcommand{\Abs}[1]{\left\lvert#1\right\rvert}

\newcommand{\Norm}[1]{\left\lVert#1\right\rVert}

\newcommand{\R}{\mathbb{R}}

\newcommand{\N}{\mathbb N}

\usepackage{bm}

\newcommand{\Ind}[1]{\mathbbm{1}\Set{#1}}

\newcommand{\independent}{\protect\mathpalette{\protect\independenT}{\perp}}
\def\independenT#1#2{\mathrel{\rlap{$#1#2$}\mkern2mu{#1#2}}}

\newcommand{\ignore}[1]{}

\newcommand{\KL}{\mathrm{KL}}

\DeclareMathOperator*{\argmin}{arg\,min}

\renewcommand{\epsilon}{\varepsilon}

\newcommand{\Unif}{\mathrm{Unif}}

\newcommand{\remove}[1]{}

\newcommand{\svert}{~\middle\vert~}
\newcommand{\ate}{\psi}
\newcommand{\sate}{\psi_{\DB}}
\newcommand{\ateP}{\psi_{\iid}}
\newcommand{\hate}{\widehat{\psi}}

\newcommand{\ite}{\psi}
\newcommand{\hite}{\widehat{\ite}}

\newcommand{\exps}[1]{\exp\Set{#1}}

\newcommand{\logs}[1]{\log\Paren{#1}}
\newcommand{\logp}[1]{\log\Paren{#1}}

\newcommand{\bz}{\mathbf{z}}
\newcommand{\bZ}{\mathbf{Z}}

\newcommand{\bY}{\mathbf{Y}}

\newcommand{\Hoeff}{\text{Hoeff}}

\newcommand{\barG}{\widebar{G}}

\newcommand{\barn}{\widebar{n}}
\newcommand{\hsigma}{\widehat{\sigma}}

\newcommand{\hmu}{\widehat{\mu}}

\newcommand{\ghate}{\widehat{\vartheta}}
\newcommand{\bghate}{\boldsymbol{\ghate}}
\newcommand{\estimator}{\widehat{\psi}}
\newcommand{\fM}{\mathfrak{M}}
\newcommand{\inv}{{-1}}
\newcommand{\minimaxrisk}{\fM\Paren{\estimator\Paren{\cP}}}
\newcommand{\EB}{\texttt{Stud}}
\newcommand{\EBCI}{\texttt{Stud-CI}}
\newcommand{\MBCR}{\texttt{MBCR}}
\newcommand{\CR}{\texttt{CR}}
\newcommand{\tBern}{\texttt{Bern}}
\newcommand{\SB}{\mathrm{SB-}}

\newcommand{\DB}{{\mathrm{DB}}}
\newcommand{\atePz}{\widebar{\psi}_{\PP_0}}
\newcommand{\atePo}{\widebar{\psi}_{\PP_1}}
\newcommand{\atePk}{\widebar{\psi}_{\PP_k}}
\newcommand{\atePM}{\widebar{\psi}_{\PP_M}}
\newcommand{\atePi}{\widebar{\psi}_{\PP_i}}
\newcommand{\atePj}{\widebar{\psi}_{\PP_j}}
\newcommand{\PP}{\mathbb P}
\newcommand{\normtp}[1]{\Norm{#1}_{L_2(\PP)}}
\newcommand{\etainvi}{{\eta^\inv(i)}}
\newcommand{\betai}{{\beta(i)}}
\newcommand{\Mirru}{{(\mathrm{u})}}
\newcommand{\Mirrl}{{(\ell)}}
\SetKwComment{Comment}{/* }{ */}

\usepackage{cancel}
\newcommand{\RR}{\mathbb R}
\newcommand{\iid}{{\mathrm{i.i.d.}}}
\newcommand{\simiid}{\overset{\iid}{\sim}}

\newcommand{\Bern}{\mathrm{Bern}}
\newcommand{\Bigoh}{{\mathcal O}}

\newcommand{\Dcal}{\mathcal D}
\newcommand{\brackone}{{(1)}}
\newcommand{\bracktwo}{{(2)}}

\newcommand{\brackM}{{(M)}}

\newcommand{\CLT}{{\mathrm{CLT}}}

\newif\ifverbose %
\verbosetrue %

\makeatletter
\newcommand*\@dblLabelI {}
\newcommand*\@dblLabelII {}
\newcommand*\@dblequationAux {}

\def\@dblequationAux #1,#2,%
    {\def\@dblLabelI{\label{#1}}\def\@dblLabelII{\label{#2}}}

\newcommand*{\doubleequation}[3][]{%
    \par\vskip\abovedisplayskip\noindent
    \if\relax\detokenize{#1}\relax
       \let\@dblLabelI\@empty
       \let\@dblLabelII\@empty
    \else %
       \@dblequationAux #1,%
    \fi
    \makebox[0.475\linewidth-1.5em]{%
     \hspace{\stretch2}%
     \makebox[0pt]{$\displaystyle #2$}%
     \hspace{\stretch1}%
    }%
    \makebox[0.05\linewidth-1.5em]{
    \hspace{\stretch1}
    \makebox[0pt]{\quad~and~\quad}
    \hspace{\stretch2}
    }
    \makebox[0.475\linewidth-1.5em]{%
     \hspace{\stretch1}%
     \makebox[0pt]{$\displaystyle #3$}%
     \hspace{\stretch2}%
    }%
    \makebox[4.15em][r]{(%
  \refstepcounter{equation}\theequation\@dblLabelI, 
  \refstepcounter{equation}\theequation\@dblLabelII)}%
  \par\vskip\belowdisplayskip
}
\makeatother

\usepackage{authblk}

\title{On Nonasymptotic Confidence Intervals for Treatment Effects in Randomized Experiments}
\author{Ricardo J. Sandoval$^\star$, Sivaraman Balakrishnan$^\dagger$, Avi Feller$^\star$, Michael I. Jordan$^{\star \ddagger}$, and Ian Waudby-Smith$^\star$

\vspace{0.4cm}
$^\star$University of California, Berkeley\\ $^\dagger$Carnegie Mellon
University\\ $^\ddagger$\'Ecole Normale Sup\'erieure \& Inria, Paris}

\affil{}
\date{\today}

\begin{document}

\maketitle

\pagenumbering{gobble}
\pagenumbering{arabic}

\begin{abstract}
We study nonasymptotic (finite-sample) confidence intervals for treatment effects in randomized experiments.
In the existing literature, the effective sample sizes of nonasymptotic confidence intervals tend to be looser than the corresponding central-limit-theorem-based confidence intervals by a factor depending on the square root of the propensity score. We show that this performance gap can be closed, designing nonasymptotic confidence intervals that have the same effective sample size as their asymptotic counterparts.  Our approach involves systematic exploitation of negative dependence or variance adaptivity (or both). We also show that the nonasymptotic rates that we achieve are unimprovable in an information-theoretic sense.

\end{abstract}

\section{Introduction}\label{section:introduction}
Randomized experiments are ubiquitous in scientific inquiry, providing researchers with a general paradigm for estimating the causal effects of interventions. The paradigm is supported by an extensive literature that derives asymptotic point estimates and confidence intervals for treatment effects.  In particular, asymptotic $(1-\alpha)$-confidence intervals $C'_n \coloneqq [L'_n, U'_n]$ generally provide the following guarantee:
\begin{equation}\label{eq:intro-asympci}
    \liminf_{n \to \infty} \PP\Paren{\ate \in C'_n} \geq 1-\alpha \mcom
\end{equation}
where $\ate$ is the average treatment effect (to be defined more formally later). Although asymptotic approximations have the virtue of simplicity and generality, they are silent on whether the approximation is an accurate one for a specific sample size. In contrast, nonasymptotic derivations aim to provide results that apply to all sample sizes.  For example, a nonasymptotic $(1-\alpha)$-confidence interval $C_n \coloneqq [L_n, U_n]$ has the following guarantee:
\begin{equation}\label{eq:intro-nonasympci}
    \forall n \in \N,~\PP\Paren{\ate \in C_n} \geq 1-\alpha \mper
\end{equation}
In this work, we focus exclusively on confidence intervals satisfying \eqref{eq:intro-nonasympci}.

Although asymptotic confidence intervals provide a weaker overall theoretical guarantee than nonasymptotic intervals, they are sometimes preferred because of the insight that they can provide into design choices such as sample sizes. For example, in the case of treatment effect estimation in causal inference, while existing asymptotic and nonasymptotic intervals both scale with their sample sizes $n$ as $1/\sqrt{n}$, they differ in terms of their \emph{effective sample sizes}. Consider a randomized experiment with a binary treatment so that the probability (the ``propensity score'') of being assigned to the treatment group is $\pi \in (0, 1/2]$.
Then, a standard central-limit-theorem-based confidence interval $C_n^{\CLT}$ and a Hoeffding-style nonasymptotic confidence interval $C_n^\Hoeff$ exhibit asymptotic scalings with respect to both $n$ and $\pi$ of

\doubleequation[eq:intro-scaling-juxtaposition-clt,eq:intro-scaling-juxtaposition-hoeff]{\Abs{C_n^\CLT} \asymp \frac{1}{\sqrt{n\pi}}}{\Abs{C_n^\Hoeff} \asymp \frac{1}{\sqrt{n \pi^2}}\mcom}
\noindent respectively. See \citet{robins2012robins}, \citet[Proposition 3]{aronow2021nonparametric}, \citet[Section 2.2]{ding2025randomization}, and \citet[Section 4.3]{tchetgen2012causal} for the pursuit of such nonasymptotic intervals with the scaling of \eqref{eq:intro-scaling-juxtaposition-hoeff}.
We can interpret 
\eqref{eq:intro-scaling-juxtaposition-clt} as scaling with an ``effective'' sample sizes of $n\pi$:
under complete randomization, exactly $n \pi$ individuals are assigned to the treatment group; under Bernoulli randomization, $n \pi$ are assigned in expectation. From this vantage point, the effective sample size of the nonasymptotic Hoeffding-style confidence interval in \eqref{eq:intro-scaling-juxtaposition-hoeff} is $n\pi^2$.

\begin{figure}
    \centering
    \begin{subcaptionblock}{0.325\textwidth}
        \centering 
        \includegraphics[width=\linewidth]{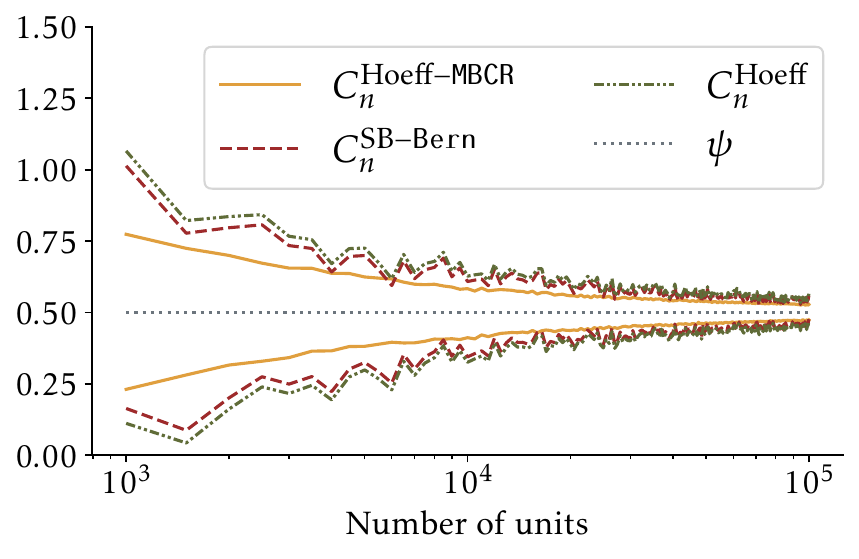}
        \caption{$\pi = 0.1$}
    \end{subcaptionblock}
    \hfill
    \begin{subcaptionblock}{0.325\textwidth}
        \centering 
        \includegraphics[width=\linewidth]{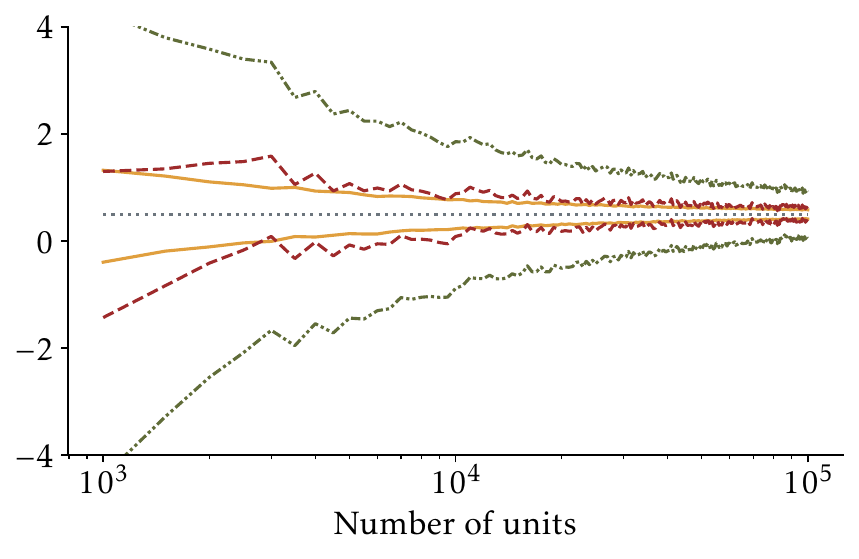}
        \caption{$\pi = 0.01$}
    \end{subcaptionblock}
    \hfill
    \begin{subcaptionblock}{0.325\textwidth}
        \centering 
        \includegraphics[width=\linewidth]{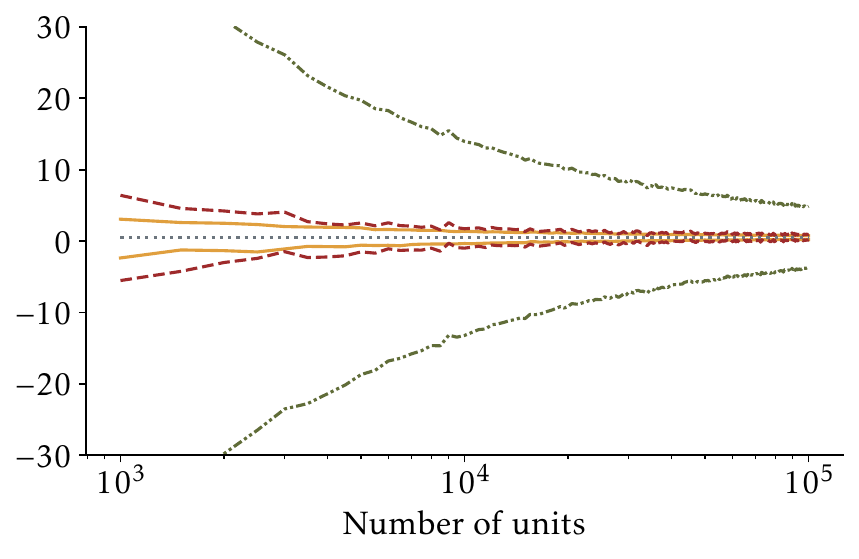}
        \caption{$\pi = 0.001$}
    \end{subcaptionblock}
    \caption{A comparison of the scaling between the nonasymptotic Hoeffding-style confidence intervals present in the literature ($C_n^{\Hoeff}$) and two of our proposed nonasymptotic confidence intervals ($C_n^{\Hoeff\mathrm{-}\MBCR}$ and $C_n^{\SB\Bern}$). The additional $1/\sqrt{\pi}$ factor in the effective sample size of the existing nonasymptotic confidence intervals leads to these intervals being quite loose---as compared to our proposed confidence intervals---in ``small-$\pi$'' regimes.}
    \label{fig:comparison-with-existing-cis}
\end{figure}

A natural question is 
whether this inferior effective sample size is a necessary property of nonasymptotic confidence interval construction or whether more delicate proof techniques can lead to nonasymptotic intervals with the same effective sample size of $n\pi$ as in the asymptotic case. In this paper,
we show that the $1/\sqrt{n \pi^2}$ scaling is in fact an artifact of existing proof techniques and a scaling of $1/\sqrt{n\pi}$ can be obtained for nonasymptotic confidence intervals for treatment effect estimation.  See \cref{fig:comparison-with-existing-cis} for a demonstration of our new confidence intervals under different values of the propensity score when compared to existing techniques. Our approach follows Hoeffding-type sub-Gaussian arguments that have been used in the literature, but it differs in the way that these arguments are employed. In particular, under complete randomization---see \cref{thm:hoeffding-mbcr}---our improvements are due to a new analysis technique that is able to exploit the appropriate (negative) dependence in the treatment effect estimation problem. Some intuition alongside a proof sketch can be found immediately following \cref{corollary:scaling-hoeffding-mbcr}. Additionally, in \cref{thm:sub-bernoulli-cis} we show that under Bernoulli randomization, the looser scaling can be ameliorated through sharper concentration inequalities for sub-Bernoulli random variables. Finally, in \cref{sec:information-theoretic-bounds} we present an information-theoretic analysis that establishes that the $1/\sqrt{n\pi}$ scaling is unimprovable.

\section{Preliminaries}\label{section:preliminaries}
We consider a setting with $n$ experimental units and two treatment arms. Each unit $i \in [n] \coloneqq \Set{1, \dots, n}$ is randomly assigned to either the treatment or control arm. We will denote as $Z_i \in \Set{0,1}$ the treatment assignment of unit $i \in [n]$, where $Z_i = 1$ denotes that the unit received the treatment and $Z_i = 0$ denotes that they received the control. For now we will remain agnostic to the randomization distribution from which $\bZ \coloneqq \Paren{Z_i}_{i=1}^n$ is sampled but we will later consider both complete and Bernoulli randomization.
Throughout, we will denote the number of units assigned to the treatment and control arms as $n_1 \coloneqq \sum_{i=1}^n Z_i$ and $n_0 \coloneqq n - n_1$, respectively.

We follow the potential outcomes framework \citep{neyman1923applications, rubin1974estimating} where each unit has a potential outcome under treatment and control. We assume that the potential outcomes are bounded and take their values in the unit interval: $(Y_i(0), Y_i(1)) \in [0,1] \times [0,1]$ for each $i \in [n]$, where $Y_i(0)$ represents the potential outcome under control and $Y_i(1)$ represents the potential outcome under treatment. One could have assumed $\Paren{Y_i(0), Y_i(1)} \in [a, b] \times [a,b]$ for arbitrary $a,b \in \R$ such that $a < b$, but without loss of generality these can be rescaled to $[0,1]$ through the transformation $x \mapsto (x - a) / (b-a)$. Our results hold under both the design-based finite population and superpopulation settings and we treat these settings jointly via the following assumption.
\begin{assumption}[Generalized potential outcome distribution]
\label{assumption:generalized-potential-outcomes}
    The potential outcomes are drawn from a distribution taking the following form: $\Paren{Y_i(0), Y_i(1)}_{i=1}^n \sim P \equiv \prod_{i=1}^n P_i$ so that $\Paren{Y_i(0), Y_i(1)} \sim P_i$ and $(Y_i(0), Y_i(1)) \in [0,1] \times [0,1]$ almost surely for each $i \in [n]$.
\end{assumption}
Assumption~\ref{assumption:generalized-potential-outcomes} reduces to the 
independent and identically distributed ($\iid$) superpopulation setting 
whenever $P_1 = \cdots = P_n$, and to the design-based finite-population 
setting whenever $P_1, \dots, P_n$ are taken to be degenerate distributions 
at the nonrandom tuples $(y_1(0), y_1(1)), \dots, \\(y_n(0), y_n(1))$.

Under the 
potential outcomes framework,
only the outcome under the assigned intervention is realized. That is, the outcome that is observed in the experiment is $Y_i \coloneqq Z_i Y_i(1) + \Paren{1 - Z_i} Y_i(0)$ for each unit $i \in [n]$. We note that in our definition of the observed outcomes we have implicitly assumed the stable unit treatment value assumption (SUTVA) --- an assumption that is sometimes referred to as ``consistency'' \citep{kennedy2024semiparametric}. This assumption posits that the outcome $Y_i$ for unit $i$ is always given by $Y_i(z)$ when they are assigned treatment level $z \in \Set{0,1}$, regardless of whether unit $j$ was assigned treatment level $z'$ for $j \neq i$ and $z' \in \Set{0,1}$.  

We focus on estimating the average treatment effect, which under the $\iid$ superpopulation and design-based settings are defined as
\begin{equation}
    \ateP \coloneqq \E_{P}\Brac{Y(1) - Y(0)} \quad \mathrm{and} \quad \sate \coloneqq \frac{1}{n} \sum_{i=1}^n \Brac{y_i(1) - y_i(0)} \mcom
\end{equation}
respectively. At times, we will deliberately be agnostic as to whether our results hold for $\ateP$ or $\sate$, in which cases we will simply write $\ate$ instead.

Identifiability of the average treatment effect relies on SUTVA in conjunction with positivity and no unmeasured confounding of the treatment assignments, assumptions that we now proceed to state formally.
\begin{assumption}[Positivity]
\label{assumption:positivity}
    We say that a treatment assignment satisfies \emph{positivity} if \\$0 < \pi < 1$, where $\pi \coloneqq \PP(Z_i = 1)$.
\end{assumption}

\begin{assumption}[No unmeasured confounding]
\label{assumption:ignorability}
    We say that a treatment assignment satisfies \emph{no unmeasured confounding} if $\Paren{Y_i(1), Y_i(0)} \independent Z_i$ for all $i \in [n]$.
\end{assumption}
Both positivity and no unmeasured confounding hold by design in our setting 
since we are considering randomized experiments.
Assumption~\ref{assumption:positivity} and Assumption~\ref{assumption:ignorability} allow us to 
express the average treatment effect as a function of the distribution
of the observed data (i.e., identify the average treatment effect). Under 
the $\iid$ superpopulation setting the average treatment effect is identified as
follows
\begin{equation}
    \ateP = \E_{\iid}\Brac{Y_i \svert Z_i = 1}
    - \E_{\iid}\Brac{Y_i \svert Z_i = 0} \mcom
\end{equation}
where we note that the expectation $\E_{\iid}[\cdot]$ is taken over both the potential outcomes and treatment assignments.
Under the design-based setting, we have that the 
average treatment effect is identified as
\begin{equation}
    \sate = \frac{1}{n} \sum_{i=1}^n \E_{\DB}\Brac{Y_i \svert Z_i = 1}
    - \E_{\DB}\Brac{Y_i \svert Z_i = 0} \mcom
\end{equation}
where the only source of randomness comes from the treatment assignments. We reflect this by using the expectation $\E_{\DB}\Brac{\cdot}$.

\subsection{Canonical randomization procedures}
In this work, we focus on providing nonasymptotic confidence intervals for the average treatment effect when using the Horvitz-Thompson estimator~\citep{horvitz1952generalization}:
\begin{equation}
    \estimator \coloneqq \frac{1}{n} \sum_{i=1}^n 
    Y_i \Paren{\frac{Z_i}{\pi} - \frac{1 - Z_i}{1 - \pi}} \mcom
\label{eq:ht-estimator}
\end{equation}
under two designs which are popular in the causal inference literature: Bernoulli randomization and complete randomization. We formally define these two randomization procedures here.
\begin{definition}[Bernoulli randomization]
\label{definition:bernoulli-randomization}
    For a treatment assignment probability $\pi \in (0, 1/2]$, the treatment assignment $Z_i$ for unit $i \in [n]$ is given by an independent and identically distributed $\mathrm{Bern}\Paren{\pi}$ random variable. Said differently, for any $(z_1, \dots, z_n) \in \{0, 1\}^n$ with $n_1$ ones and $n_0$ zeros,
    \begin{equation}
        \PP_{\emph{\tBern}}\left ((Z_1, \dots, Z_n) = (z_1, \dots, z_n)\right ) = \pi^{n_1}(1-\pi)^{n_0}.
    \end{equation}
\end{definition}

\begin{definition}[Complete randomization]
\label{definition:complete-randomization}
    Fix the number of treated units, $n_1 \in \N$, and control units, $n_0 \in \N$, such that $n = n_1 + n_0$. Under  \emph{complete randomization}, the treatment assignment vector $\Paren{Z_1, \dots, Z_n}$ has the 
    distribution given by
    \begin{align}
        \PP_{\emph{\CR}}\Paren{\Paren{Z_1, \dots, Z_n} = \Paren{z_1, \dots, z_n}} = \frac{1}{{n \choose n_1}} \mcom
    \end{align}
    where $\Paren{z_1, \dots z_n}$ is any permutation of the vector of $n_1$ ones followed by $n_0$ zeros.
\end{definition}

The construction of our confidence intervals under complete randomization will depend on a tool we call ``mini-batch complete randomization.'' This is a randomization procedure that is distributionally equivalent to complete randomization but which allows us to better exploit the negative dependence between treatment assignments.

\subsection{Mini-batch complete randomization}\label{section:mini-batch-complete-randomization}

At a high level, mini-batch complete randomization comprises a series of two random permutations. The first permutation is over all units, which induces a random ordering that is then used to group units together. Then, within each group we randomly assign units to the treatment and control arms by randomly permuting that group's treatment assignment vector. As alluded to above, this randomization scheme is marginally identical to complete randomization (see \cref{proposition:mbcr-cr-equivalence}) with the advantage of supplying additional bookkeeping of the negative dependence between treatment assignments.
We now proceed to describe mini-batch complete randomization formally.

We define the marginal propensity score $\pi$ as
\begin{align}
    \pi \coloneqq \frac{n_1}{n} \leq \frac{1}{2} \mper
\end{align}
To perform mini-batch complete randomization, we first compute the size of the batches, or groups, as
\begin{align}
\label{eq:group-size}
    G \equiv G\Paren{\pi} \coloneqq \left\lceil \frac{1}{\pi} \right\rceil \geq 2 \mper
\end{align}
For example, if $\pi = \flatfrac{1}{K}$ for an integer $K \geq 2$, then $G = K$, and if $\pi = \flatfrac{2}{7}$, then $G = 4$, and so on.
Let $T$ denote the total number of groups of size $G$:
\begin{align}
    T \coloneqq \begin{cases}
      n_1 & \text{if}~\left \lceil 1/\pi \right \rceil = 1/\pi \mcom\\
      n_1-1 & \text{if}~n-(n_1 - 1) G \geq 2 \mcom\\
      n_1-2 & \text{otherwise}.
    \end{cases}
\end{align}
Let $\widebar G$ be the size of the final group given by
\begin{equation}
  \widebar G = n - T G,
\end{equation}
so that the number of units randomized to one of the treatment or control groups is given by $TG + \widebar G = n$.
Let $\widebar n_1$ be the number of units to be assigned to treatment in the final group, given by
\begin{equation}
    \widebar n_1 \coloneqq \begin{cases}
      0 & \text{if}~\left \lceil 1/\pi \right \rceil = 1/\pi \mcom\\
      1 & \text{if}~n-(n_1 - 1) G \geq 2 \mcom\\
      2 & \text{otherwise}.
    \end{cases}
\end{equation}
Notice that at each of the first $T$ rounds we  assign exactly one unit to treatment and at the final round we assign $\widebar n_1$ to the treatment arm. This ensures that exactly $n_1$ units are assigned to the treatment arm altogether.

In what follows, let $S(N)$ be the group of permutations on $\{1, \dots, N\}$ for any positive integer $N$. Let $a \equiv (a_1, \dots, a_n)$ be a vector containing $n_0$ zeros and $n_1$ ones arranged as a single $1$ followed by $G-1$ zeros, then a single $1$ followed by $G-1$ zeros and so on $T$ times, and finally $\widebar n_1$ ones followed by $\widebar G - \widebar n_1$ zeros. The procedure to construct $a$ is best illustrated via examples. Suppose in particular that $n = 9$ with $n_1 = 3$ and $n_0 = 6$. Then,
\begin{equation}
  a = (1, 0, 0, 1, 0, 0, 1, 0, 0).
\end{equation}
Similarly, if $n = 9$ with $n_1 = 4$ and $n_0 = 5$, then 
\begin{equation}
  a = (1, 0, 0, 1, 0, 0, 1, 1, 0),
\end{equation}
and so on. Given an organized vector of pre-randomization treatment allocations, mini-batch complete randomization---stated formally in \cref{algorithm:mbcr} found in \cref{sec:mbcr-algorithm}---is nothing but a particular random permutation of $a$. The permutation first consists of a random permutation of within-subgroup treatment assignments, denoted by $\beta$ in \cref{algorithm:mbcr}, and then a uniformly random permutation of the entire $n$-vector, denoted by $\eta$ in \cref{algorithm:mbcr}, to ultimately yield their composition $\rho \coloneqq \eta \circ \beta$. The vector of treatment assignments $\bZ \equiv (Z_1, \dots, Z_n)$ is $\bZ = (a_{\rho(1)} , \dots, a_{\rho(n)})$.

As discussed above, mini-batch complete randomization is distributionally equivalent to complete randomization in the sense that both satisfy \cref{definition:complete-randomization}. We state this result formally in the following proposition.

\begin{proposition}
\label{proposition:mbcr-cr-equivalence}
    Let $\PP_{\emph{\MBCR}}$ and $\PP_{\emph{\CR}}$ be the probability distributions induced by mini-batch complete randomization and complete randomization, respectively. Then, for any permutation $(z_1, \dots, z_n)$ of $n_1$ ones and $n_0$ zeros,
    \begin{align}
        \PP_{\emph{\MBCR}}\Paren{\Paren{Z_1, \dots, Z_n} = \Paren{z_1, \dots, z_n}} = \PP_{\emph{\CR}}\Paren{\Paren{Z_1, \dots, Z_n} = \Paren{z_1, \dots, z_n}} = \frac{1}{{n \choose n_1}} \mper
    \end{align}
\end{proposition}
The proof can be found in \cref{section:equivalence-mbcr-cr}.  It consists of showing that after conditioning on the within group treatment assignment, $\beta$, the assignment probability under mini-batch complete randomization is the same as that of complete randomization.

\begin{remark}[Connection to stratified randomized experiments]
    Our mini-batch complete randomization procedure is reminiscent of 
    stratified randomized experiments---or matched-pairs experiments in their
    most ``extreme'' case. Rather than optimizing for covariate balance 
    in the treatment and control groups, the objective of mini-batch complete
    randomization is to have a procedure that imposes enough structure that enables
    us to adequately exploit the negative dependence amongst treatment assignments.
\end{remark}

Under mini-batch complete randomization different units may have different conditional propensity scores; thus, we can no longer use the standard Horvitz-Thompson estimator. Instead, for the first $T$ groups---which we denote as $g_t$ for $t = 1, \dots , T$---the propensity score conditional on the unit-wide permutation $\eta$ is $\flatfrac{1}{G}$, because only one unit in each group of size $G$ is assigned to the treatment. Meanwhile, for the last group $\widebar g$ the propensity score is $\flatfrac{\widebar n_1}{\widebar G}$, because $\widebar n_1$ units of the $\widebar G$ are assigned to the treatment. Our Horvitz-Thompson estimator thus takes the following form.

\begin{definition}[Horvitz-Thompson estimator under mini-batch complete randomization]
\label{definition:mbcr-ht-estimator}
The \emph{Horvitz-Thompson estimator} under the mini-batch complete randomization design is
    \begin{align}
    \label{eq:mbcr-horvitz-thompson-estimator}
        \estimator' \coloneqq \frac{1}{n} \sum_{t=1}^T \sum_{i \in g_t}  Y_{\eta^{-1}(i)} \left ( \frac{Z_{\beta(i)}}{1/G} - \frac{(1-Z_{\beta(i)})}{1-1/G} \right ) + \frac{1}{n} \sum_{i\in \widebar g}  Y_{\eta^{-1}(i)} \left ( \frac{Z_{\beta(i)}}{1/\widetilde G} - \frac{(1-Z_{\beta(i)})}{1-1/\widetilde G} \right ) \mcom
    \end{align}
    where $\widetilde G \coloneqq \flatfrac{\widebar G}{\widebar n_1}$.
\end{definition}
Notice how the indices for the treatment assignments and the observed outcomes are no longer the usual ``$i$''; rather, they are $\beta(i)$ and $\eta^{-1}(i)$, respectively. These are necessary transformations because we are now keeping track of the groups and units within each group. Furthermore, in the case that $T = n_1$ the Horvitz-Thompson estimator under mini-batch complete randomization \eqref{eq:mbcr-horvitz-thompson-estimator} is equal to the standard Horvitz-Thompson estimator from \cref{eq:ht-estimator}. 
For simplicity, throughout the paper we will often refer to the estimator from \eqref{eq:mbcr-horvitz-thompson-estimator} as the Horvitz-Thompson estimator, but we reiterate that it is only equal to the usual Horvitz-Thompson estimator in settings in which $\pi = 1/K$ for an integer $K \geq 2$.

As we state in the following proposition, whose proof is provided in \cref{section:unbiased-ht}, the Horvitz-Thompson estimator from \eqref{eq:mbcr-horvitz-thompson-estimator} produces an unbiased estimate of the average treatment effect.
\begin{proposition}
\label{proposition:unbiased-ht}
    Under mini-batch complete randomization the Horvitz-Thompson estimator
    \eqref{eq:mbcr-horvitz-thompson-estimator} is a conditionally and
    maringally unbiased estimator of the sample average treatment effect
    \begin{equation}
        \E_{\DB}\Brac{\hate' \svert \eta} = \E_{\DB}
        \Brac{\hate'} = \frac{1}{n} \sum_{i=1}^n [y_i(1) - y_i(0)] \mper
    \end{equation}
    Moreover, under mini-batch complete randomization the Horvitz-Thompson
    estimator \eqref{eq:mbcr-horvitz-thompson-estimator} is a conditionally and
    maringally unbiased estimator of the population average treatment effect
    \begin{equation}
        \E_{\iid}\Brac{\hate' \svert \eta} = \E_{\iid}\Brac{\hate'} = \E_{P}\Brac{Y(1) - Y(0)} \mper
    \end{equation}
\end{proposition}

Now that we have introduced mini-batch complete randomization, 
we are ready to present confidence intervals that enjoy the $1/\sqrt{n\pi}$ scaling.

\section{Confidence Intervals with the Optimal Effective Sample Size}\label{section:confidence-intervals}

In the first part of this section we present nonasymptotic confidence intervals that enjoy the $n\pi$ effective sample size alluded to in the introduction. We show that such confidence intervals can be constructed under complete, mini-batch complete, and Bernoulli randomization, and discuss the key ideas behind their constructions. In the latter part of this section we present a variance-adaptive confidence interval that is reminiscent of classical Studentized variance-adaptive confidence intervals but is able to handle instances in which the data are not independent and identically distributed.

\subsection{Hoeffding-style and sub-Bernoulli confidence intervals}

The first result we present is a Hoeffding-style confidence interval under (mini-batch) complete randomization. Its construction depends on a bound on the moment generating function of bounded random variables due to \citet{hoeffding_probability_1963}, commonly referred to as Hoeffding's lemma. The use of Hoeffding's lemma to construct confidence intervals for the average treatment effect has appeared numerous times in the literature~\citep{robins2012robins, aronow2021nonparametric, ding2025randomization, tchetgen2012causal} but in all cases en route to a suboptimal effective sample size of $n\pi^2$.
\begin{theorem}
\label{thm:hoeffding-mbcr}
    Let $\hate'$ be the Horvitz-Thompson estimator as defined in \cref{definition:mbcr-ht-estimator} and define the interval
    \begin{equation}
    \label{eq:hoeffding-mbcr-general-form}
        C_n^{\emph{\Hoeff}\mathrm{-}\emph{\MBCR}} \coloneqq \left [ \hate' \pm c_n^{\emph{\MBCR}}(\pi) \sqrt{\frac{2\log(2/\alpha)}{n}} \right ]
    \end{equation}
    where the constant $c_n^{\emph{\MBCR}}(\pi)$ takes the following form:
    \begin{equation}
    \label{eq:hoeffding-mbcr-constant}
        c_n^{\emph{\MBCR}}(\pi) \coloneqq \sqrt{\frac{T G^2 + \widebar G^2}{n}} \mper
    \end{equation}
    Then, under mini-batch complete randomization we have that $C_n^{\emph{\Hoeff}\mathrm{-}\emph{\MBCR}}$ is a nonasymptotic $(1-\alpha)$-confidence interval for the average treatment effect under the design-based and i.i.d. superpopulation settings. More formally, for all $n \in \N$ we have that
    \begin{equation}
        \PP_{\DB}\Paren{\sate \in C_n^{\emph{\Hoeff}\mathrm{-}\emph{\MBCR}}} \geq 1-\alpha
        \quad\mathrm{and}\quad
        \PP_{\iid}\Paren{\ateP \in C_n^{\emph{\Hoeff}\mathrm{-}\emph{\MBCR}}} \geq 1-\alpha .
    \end{equation}
\end{theorem}
We make two remarks. First, notice that the confidence interval in \eqref{eq:hoeffding-mbcr-general-form} is centered around the slightly modified Horvitz-Thompson estimator from \eqref{eq:mbcr-horvitz-thompson-estimator}. When the propensity score takes the form $\pi = 1/K$ for some integer $K \geq 2$---which corresponds to the instance when $\widebar{n}_1 = 0$---the confidence interval from \eqref{eq:hoeffding-mbcr-general-form} is centered around the unmodified Horvitz-Thompson estimator from \eqref{eq:ht-estimator}. That is, when $\pi = 1/K$ for an integer $K \geq 2$, the confidence interval $C_n^{\Hoeff\mathrm{-}\MBCR}$ can be written without any reference to mini-batch complete randomization (and indeed without running the procedure at all).
Secondly, the constant in \eqref{eq:hoeffding-mbcr-constant} depends on the size and number of batches that result from the mini-batch complete randomization procedure. In the following corollary, we either write or upper bound the constant from \eqref{eq:hoeffding-mbcr-constant} purely in terms of the propensity score $\pi$.
\begin{corollary}
\label{corollary:scaling-hoeffding-mbcr}
The mini-batch complete randomization constant $c_n^{\emph{\MBCR}}(\pi)$ in \eqref{eq:hoeffding-mbcr-constant} can be given or upper bounded by
\begin{alignat}{3}
    &c_n^{\emph{\MBCR}}(\pi)&\  =\ & 1/\sqrt{\pi} \quad &\text{if $\widebar n_1 = 0$}& \label{eq:hoeffding-mbcr-scaling-case-0}\\
    &c_n^{\emph{\MBCR}}(\pi)&\ \leq\ & (1+\pi) / \sqrt{\pi} \quad &\text{if $\widebar n_1 = 1$}&\\
    &c_n^{\emph{\MBCR}}(\pi)&\ \leq\ & \sqrt{\flatfrac{\Paren{1 + \pi}^2}{\pi} + 2\flatfrac{\Paren{1/\pi + 1}^2}{n}} \quad &\text{if $\widebar n_1 = 2$}& \mper
\end{alignat}
\end{corollary}
\noindent
In all three cases, $\widebar n_1 \in \{0, 1,2\}$, the width of the confidence interval  scales as
\begin{equation}\label{eq:hoeffding-mbcr-scaling}
\Abs{C_n^{\Hoeff\mathrm{-}\MBCR}} \asymp \sqrt{\frac{2 \log (2/\alpha)}{n\pi}},
\end{equation}
in the large-$n$, small-$\pi$ regime, or more formally,
\begin{equation}
 \lim_{\pi \to 0} \lim_{n \to \infty}\sqrt{n \pi} \Abs{C_n^{\Hoeff\mathrm{-}\MBCR}} = \sqrt{2 \log (2/\alpha)}.
\end{equation}

The proofs of \cref{thm:hoeffding-mbcr,corollary:scaling-hoeffding-mbcr} can be found in \cref{section:proof-hoeffding-mbcr}. The derivation of \cref{thm:hoeffding-mbcr} showcases a new proof technique enabled by the construction of mini-batch complete randomization. Given its simplicity as well as its centrality to the results of the present section, we give a brief proof sketch here.

\begin{proof}[Proof sketch for \cref{thm:hoeffding-mbcr} and \cref{corollary:scaling-hoeffding-mbcr}]
  For the sake of simplicity, we consider the case when $\pi = 1/K$ for an integer $K \geq 2$ (equivalently, when $\widebar n_1 = 0$), but the full proof in \cref{section:proof-hoeffding-mbcr} contains details for the more general case. Consider the exponential random variable
  \begin{equation}\label{eq:proof sketch M_n}
    M_n = \exp \left \{ \lambda \sum_{t=1}^T \sum_{i \in g_t} \left [ Y_{\eta^{-1}(i)} \left ( \frac{Z_{\beta(i)}}{1/G} - \frac{1-Z_{\beta(i)}}{1-1/G} \right ) - \ate \right ] - T \frac{\lambda^2 \sigma^2}{2} \right \};\quad \lambda \in \RR,
  \end{equation}
  where $\sigma^2 \coloneqq G^2$ and we recall that $T = n\pi$ in the current setting. Notice that by construction of the sub-group treatment assignment $\beta$, for each $t \in [T]$ it holds that $Z_{\beta(i)} = 1$ for exactly one $i \in g_t$. As such, we have that for each $t \in [T]$,
  \begin{equation}
    -G \leq \sum_{i\in g_t} Y_{\eta^{-1}(i)} \left ( \frac{Z_{\beta(i)}}{1/G} - \frac{1-Z_{\beta(i)}}{1-1/G} \right ) \leq G\quad \text{almost surely.}
  \end{equation}
  By Hoeffding's lemma \citep{hoeffding_probability_1963}, we have that conditionally on the unit-wide permutation $\eta$, $M_n$ is nonnegative with mean at most 1 for any $\lambda \in \RR$ when $\sigma^2$ is taken to be $(2G)^2/4 = G^2$. Applying Markov's inequality to obtain $\PP (M_n \geq 1/\alpha) \leq \alpha$ for any $\alpha \in (0, 1)$ and setting $\lambda= \sqrt{2\log(1/\alpha) / (TG^2)}$, it holds that for any $\alpha \in (0, 1)$,
  \begin{equation}
    \PP \left ( \frac{1}{n}\sum_{t=1}^T \sum_{i \in g_t} Y_{\eta^{-1}(i)} \left ( \frac{Z_{\beta(i)}}{1/G} - \frac{1-Z_{\beta(i)}}{1-1/G} \right ) - \ate \geq G \sqrt{\frac{2 \log(1/\alpha) T}{n^2}} \right )  \leq \alpha.
  \end{equation}
  Identifying $G$ with $1/\pi$ and $T$ with $n\pi$, we see that the right-hand side of the inequality inside the above probability becomes
  \begin{equation}\label{eq:mbcr-magic}
    \sqrt{\frac{2 \log (1/\alpha) n \pi}{n^2 \pi^2}} = \sqrt{\frac{2 \log (1/\alpha)}{n \pi}}.
  \end{equation}
  Repeating the above derivation with negatives of the inverse-probability-weighted summands and $-\psi$ combined with a union bound completes the sketch of the proof.
\end{proof}
Let us now contrast the above proof technique with what has previously appeared in the related literature. A typical Hoeffding-style proof (including those found in \citep{robins2012robins,tchetgen2012causal,aronow2021nonparametric,ding2025randomization}) would have proceeded by analyzing the exponential random variable
\begin{equation}\label{eq:naive exponential concentration e-value}
  M_n' := \exp \left \{ \lambda \sum_{i=1}^n \left [Y_i \left ( \frac{Z_i}{\pi} - \frac{1-Z_i}{1-\pi} \right ) - \psi \right ] - n \frac{\sigma^2 \lambda^2}{2} \right \};\quad \lambda \in \RR
\end{equation}
for $\sigma^2 > 0$ to be chosen shortly,
and exploiting the fact that the summands $Y_i (Z_i / \pi - (1-Z_i) / (1-\pi))$ for all $i\in[n]$ lie in the interval $[-1/(1-\pi), 1/\pi]$ with probability one. By Hoeffding's lemma \citep{hoeffding_probability_1963}, $M_n'$ is also nonnegative with mean at most 1 for any $\lambda \in \RR$ when $\sigma^2$ is taken to be $[1/(1-\pi) + 1/\pi]^2 / 4$. Using similar ideas to those found in the proof sketch above but with $\lambda$ set to $\sqrt{2 \log(1/\alpha) / (n \sigma^2)}$, one would then arrive at the concentration inequality
\begin{equation}
  \PP \left ( \frac{1}{n} \sum_{i=1}^n Y_i \left ( \frac{Z_i}{\pi} - \frac{1-Z_i}{1-\pi} \right ) - \ate \geq \left ( \frac{1}{1-\pi} + \frac{1}{\pi} \right )\sqrt{\frac{\log(1/\alpha)}{2n}} \right ) \leq \alpha.
\end{equation}
A confidence interval derived from the above would clearly have an effective sample size of $n\pi^2$ instead of $n\pi$.
\begin{remark}\label{remark:on complete randomization}
  The works of \citet{robins2012robins,aronow2021nonparametric,ding2025randomization} only consider Bernoulli randomization, but their proofs are immediately amenable to complete randomization without any change to the confidence interval using the ideas found in \citet[Theorem 15]{chi2022multiple} or \citet{joag1983negative}. Such an amendment would not improve the effective sample size.
\end{remark}
It is not immediately obvious that the construction of mini-batch complete randomization and the exponential random variable in \eqref{eq:proof sketch M_n} would lead to sharper concentration than \eqref{eq:naive exponential concentration e-value}, since the sample size (insofar as concentration of measure is concerned) in \eqref{eq:naive exponential concentration e-value} is larger by a factor of $1/\pi$. Nevertheless, the almost sure bounds of $[-1/(1-\pi), 1/\pi]$ of the summands in \eqref{eq:naive exponential concentration e-value} are sufficiently conservative that a smaller concentration sample size of $n\pi$ with the combined almost sure bounds of $G = 1/\pi$ of the random variables in \eqref{eq:proof sketch M_n} leads to inequalities that are sharper by a factor of $\sqrt{2\pi}$.

Under Bernoulli randomization we are no longer able to exploit the dependence between the assignments induced by (mini-batch) complete randomization. Nevertheless, as we show in \cref{thm:sub-bernoulli-cis}, we can still construct confidence intervals under Bernoulli randomization that enjoy the optimal scaling. These confidence intervals are based on a sub-Bernoulli concentration inequality and depend on the following cumulant generating function (see \citet[Lemma 1]{hoeffding_probability_1963} and \citet{howard2021time}):
\begin{equation}
    \gamma_B(\lambda) \coloneqq \log \left ( \frac{b}{b-a} e^{\lambda a} - \frac{a}{b-a} e^{\lambda b} \right ) \mcom
    \label{eq:sub-bernoulli-cgf}
\end{equation}
for scalars $a < b$.
\begin{remark}\label{remark:bernoulli}
  The term ``Bernoulli'' in ``sub-Bernoulli'' has nothing to do with Bernoulli randomization. Instead, sub-Bernoulli random variables are those whose cumulant generating functions are smaller than those of a Bernoulli random variable. To avoid confusion going forward, the term ``Bernoulli'' will always be preceded by the modifier ``sub-'' when referring to properties of a random variable's cumulant generating function.
\end{remark}
The sub-Bernoulli confidence intervals under both Bernoulli and mini-batch complete randomization explicitly depend on \eqref{eq:sub-bernoulli-cgf}. Under Bernoulli randomization we set $a_{\tBern} \coloneqq -1/(1-\pi) - 1$ and $b_{\tBern} \coloneqq 1/\pi + 1$, while under mini-batch complete randomization we set $a_{\MBCR} \coloneqq -2G$ and $b_{\MBCR} \coloneqq 2G$.
Furthermore, we define the parameters
\begin{equation}
    \lambda_{\MBCR} \coloneqq \sqrt{\frac{2\logp{2/\alpha}}{T G^2}} \quad\mathrm{and}\quad \lambda_{\tBern} \coloneqq \sqrt{\frac{2\logp{2/\alpha}}{n \Paren{\frac{1}{1-\pi} + 1} \Paren{\frac{1}{\pi} + 1}}} \mcom
\end{equation}
under mini-batch complete randomization and Bernoulli randomization, respectively. For expositional simplicity, we present our sub-Bernoulli confidence interval under mini-batch complete randomization only for the case in which $\pi = 1/K$ for some integer $K \geq 2$.

\begin{theorem}[Sub-Bernoulli confidence intervals]
\label{thm:sub-bernoulli-cis}
    Under Bernoulli and mini-batch complete randomization
    we have that
    \begin{equation}
        C_n^{\mathrm{SB}-*} \coloneqq \Brac{\estimator \pm \frac{\logp{2/\alpha} + \kappa_n^*(\pi)}{n\lambda_*}}
    \end{equation} 
    forms a $(1-\alpha)$-confidence interval for the average treatment effect, where $* \in \Set{\emph{\texttt{Bern}}, \emph{\MBCR}}$ and
    \begin{equation}
        \kappa_n^{\emph{\tBern}} (\pi) \coloneqq n \gamma_{B}(\lambda_{\emph{\tBern}})
         \mcom \quad\text{and}\quad  \kappa_n^{\emph{\MBCR}}(\pi) \coloneqq T \logs{\frac{1}{2} e^{-2G \lambda_{\emph{\MBCR}}} + \frac{1}{2} e^{2G \lambda_{\emph{\MBCR}}}} \mper
    \end{equation}
    Furthermore, in the large-$n$, small-$\pi$ regime the confidence intervals scale as
    \begin{equation}
        \Abs{C_n^{\mathrm{SB}-\emph{\tBern}}} \asymp \sqrt{\frac{4\logp{2/\alpha}}{n\pi}}
         \quad\mathrm{and}\quad 
         \Abs{C_n^{\mathrm{SB}-\emph{\MBCR}}} \asymp \sqrt{\frac{8\logp{2/\alpha}}{n\pi}}
    \end{equation}
    under Bernoulli randomization and mini-batch complete randomization, respectively.
\end{theorem}
The scaling that we have obtained involves an effective sample size of $n \pi$ and improves upon existing nonasymptotic confidence intervals for the Bernoulli randomization setting~\citep{robins2012robins,aronow2021nonparametric,ding2025randomization}. The improvement in this case comes from using a better-suited concentration inequality for sub-Bernoulli random variables. 

Notice that the scaling of the sub-Bernoulli confidence interval under Bernoulli randomization enjoys a better constant than its counterpart under mini-batch complete randomization. Nonetheless, the scalings of these confidence intervals are larger than that of our Hoeffding-style confidence interval under mini-batch complete randomization \eqref{eq:hoeffding-mbcr-scaling} by multiplicative constants. 
Our recommendation is thus to use the sub-Bernoulli confidence interval from \cref{thm:sub-bernoulli-cis} when one is working under Bernoulli randomization, and to use the Hoeffding-style confidence interval when one is working with mini-batch complete randomization. When one is working under complete randomization, however, our recommendation is more nuanced. In settings in which the propensity score is $\pi = 1/K$ for some integer $K \geq 2$, one should use the Hoeffding-style confidence interval from \cref{thm:hoeffding-mbcr}---as in this case mini-batch complete randomization and complete randomization yield identical Horvitz-Thompson estimators, allowing one to directly use the results under mini-batch complete randomization. Meanwhile, when the propensity score is $\pi \neq 1/K$ for an integer $K \geq 2$, we recommend that one uses the sub-Bernoulli confidence interval---as bounds derived under Bernoulli randomization are valid bounds under complete randomization (e.g., see \citep[Section 6]{hoeffding_probability_1963}).

Together, \cref{thm:hoeffding-mbcr,thm:sub-bernoulli-cis} show that the $\flatfrac{1}{\sqrt{n \pi^2}}$ scaling that has appeared in the literature thus far is an artifact of the analysis used to derive the confidence intervals, and not a limitation of the nonasymptotic regime \emph{per se}. In what follows, we show that we can also obtain variance-adaptive confidence intervals for the average treatment effect under both the design-based finite population and superpopulation settings.

\subsection{Studentized variance-adaptive confidence intervals}\label{section:empirical-bernstein-ci}
In the previous sections we showed that the width of nonasymptotic confidence 
intervals for the average treatment effect can scale at a parametric 
rate with respect to the effective sample size. A natural upgrade to the 
previous confidence intervals is one whose width is \emph{variance-adaptive}, 
in the sense that the intervals may scale not only with the effective sample size 
but also with the empirical variance induced by the potential outcomes. 
Variance adaptivity is a property common in ``Bernstein'' or
``empirical-Bernstein''-type confidence intervals, but such adaptivity has not informed the design of confidence intervals for the design-based setting or under complete randomization.
In this subsection we fill this gap by presenting ``Studentized'' 
variance-adaptive confidence intervals for both  Bernoulli and mini-batch complete randomization. Importantly, the 
Studentized intervals we introduce are applicable even when the potential outcomes are non-$\iid$, and their width scales as a function of the 
empirical variance of the random variables.

Variance-adaptive confidence intervals are common in the statistics and machine learning literatures~\citep{audibert2007tuning,maurer2009empirical,bardenet2015concentration,howard2021time,waudby2024estimating,waudby2024anytime}. 
However, these intervals are only be applicable to superpopulation settings and/or Bernoulli designs (with the exception of \citet{bardenet2015concentration} which considers Bernstein bounds under without-replacement sampling). \citet{howard2021time} and \citet{waudby2024anytime} provide confidence sequences for the 
``running'' average treatment effect which can be instantiated for the sample average treatment effect $\sate$, but those bounds scale with the sample size no faster than $\sqrt{\log \log (n) / n}$. 
\citet{waudby2024estimating} derive another empirical-Bernstein confidence sequence, 
but a naive application of their methods to our setting would require identical 
individual treatment effects across all units, and would thus not yield valid 
confidence intervals for $\sate$.

In order to construct our Studentized confidence interval we adapt the so-called 
``mirroring trick''~\citep{thomas2015high,waudby2024anytime}, to the case
of an estimator containing differences of inverse-probability-weighted observations. The essential idea behind this trick is to exploit the structure of importance weights---in this case, $Z_i / \pi$ and $(1-Z_i) / (1-\pi)$ for $i \in [n]$---to derive bounds that will scale with shrinking propensity scores only if the empirical variances of inverse probability weighted estimators do as well. For example, consider the bound derived in \cref{corollary:scaling-hoeffding-mbcr}, noting that the width of this confidence interval scales inversely with $\sqrt{\pi}$ regardless of what values the potential outcomes take. On the other hand, the Horvitz-Thompson pseudo-outcome for unit $i\in[n]$,
\begin{equation}
    \hite_i = Y_i \left ( \frac{Z_i}{\pi} - \frac{1-Z_i}{1-\pi} \right ),
\end{equation}
could have a small or even near-zero variance if $Y_i(1) \approx Y_i(0) \approx 0$. The mirroring trick (to be described in more detail shortly) is employed so that confidence intervals whose empirical variance is very small need not grow with $1/\sqrt{\pi}$.

In what follows, we will make use of the following cumulant-generating-like function (in the sense 
of~\citep{howard2021time}) for random variables that have one-sided sub-exponential tails:
\begin{equation}
\label{eq:sub-exponential-cgf}
    \gamma_{E,c}(\lambda) \coloneqq c^{-2}\Paren{-\logs{1-c\lambda} - c\lambda}
    \mcom
\end{equation}
for values of $\lambda \in \R_{\geq 0}$ and $c \in \R_{>0}$ which we instantiate shortly.
We will construct lower and upper confidence intervals that are \emph{a priori} anchored at different estimators (though as we discuss in \cref{remark:the form of mirrored estimators}, they coincide under mini-batch complete randomization).
Let us first recall the ``standard'' Horvitz-Thompson estimator: 
\begin{equation}
\label{eq:lower-estimator}
    \estimator_{\etainvi}^{\Mirrl} \equiv \estimator_{\etainvi} = Y_{\etainvi}\Paren{\frac{Z_{\beta(i)}}{1/G} - \frac{1 - Z_{\beta(i)}}{1-1/G}} \mcom
\end{equation}
where both $\eta^{-1}(\cdot)$ and $\beta(\cdot)$ are identity maps under 
Bernoulli randomization, and they are the unit-wide and within-sub-group 
random permutations under mini-batch complete randomization. Under 
Bernoulli randomization we let $1/G = \pi$. We will use the estimator from \eqref{eq:lower-estimator} to construct the lower confidence set. To construct the upper confidence set we consider the following ``mirrored'' estimator \citep{thomas2015high,waudby2024anytime}
\begin{equation}
\label{eq:upper-mirrored-estimator}
    \estimator_{\etainvi}^{\Mirru} \coloneqq \Paren{Y_{\etainvi} - 1}\Paren{\frac{Z_{\beta(i)}}{1/G} - \frac{1 - Z_{\beta(i)}}{1-1/G}} \mcom
\end{equation}
which is an unbiased estimator for $\ite_i = \E_{P_i}\Brac{Y_i(1) - Y_i(0)}$. To see where the adjective ``mirrored'' comes from, notice how the ``standard'' Horvitz-Thompson estimator from \eqref{eq:lower-estimator} takes it values from the interval $[-1/(1-1/G), 1/G]$ while our ``mirrored'' estimator from \eqref{eq:upper-mirrored-estimator} takes it values from the interval $[-1/G, 1/(1-1/G)]$. That is, the interval from which \eqref{eq:upper-mirrored-estimator} takes its values is the reflection from our original estimator---allowing our intervals to have an even better dependence on $\pi$ when making use of the cumulant generating function from \eqref{eq:sub-exponential-cgf}.
We have not explicitly written the case of mini-batch complete randomization 
for which units belong to group $\widebar g$ and whose propensity score is 
$\flatfrac{1}{\widetilde G}$, but analogous ``mirrored'' estimators are used for 
them (we point the reader to \cref{section:proof-empirical-bernstein} for their definition).

We refer to the confidence interval from \cref{theorem:cross-fit-empirical-bernstein}
as a \emph{Studentized} confidence interval because it depends on estimates of
the variance of the random variables. To construct such a confidence interval
we use a cross-fitting method that provides the necessary estimates.
The cross-fitting method we devise is different from that typically used in doubly 
robust causal inference 
\citep{robins2008higher,chernozhukov2018double,kennedy2024semiparametric}, 
and it proceeds as follows.
We split the $n$ data points into two splits of size 
$m_1 \coloneqq \lfloor \flatfrac{\widebar T}{2} \rfloor$ and 
$m_2 \coloneqq \widebar T - m_1$, where $\widebar T \in \N$ represents the total 
number of ``groups'' which we denote as $g_t$ for $t \in [\widebar T] \coloneqq \Set{1, \dots, \widebar T}$. 
Under Bernoulli randomization $g_t = \Set{t}$ and $\widebar T = n$ while under 
mini-batch complete randomization $g_t$ is the set of units who belong to batch 
$t$ and we let $\widebar T = T + \Ind{\widebar G > 0}$. To simplify the 
exposition of our results, we will use the convention that $g_{\widebar T}$ 
will denote the group $\widebar g$ whenever $\Ind{\widebar G > 0} = 1$. 
We let $\ghate_t^* \coloneqq \sum_{i \in g_t} \hite_\etainvi^*$ for $* \in \Set{\Mirrl, \Mirru}$ 
be the sum of the individual treatment effect estimators for those units in $g_t$. 

We use the different splits to construct estimates of the average and variance of
the ``group-wise'' treatment effect estimators $\ghate_t$, and we compute these estimates as
follows:
\begin{align}
    \hmu_{t-1}^{*}\brackone \coloneqq \frac{\sum_{j=m_1+1}^{\widebar T}  \ghate_j^* + \sum_{j=1}^{t-1} \ghate_j^*}{m_2 + t-1}\mcom \quad V_{m_1}^*\brackone \coloneqq \sum_{t=1}^{m_1}\Paren{\ghate_t^* - \hmu_{t-1}^*\brackone}^2 \mcom \quad\text{and}\quad \hsigma_1^{* 2} \coloneqq \frac{V_{m_1}^*\brackone}{m_1} \mcom
\label{eq:main-paper-emp-Bernstein-avg-variance-definition}
\end{align}
for $* \in \Set{\Mirrl, \Mirru}$ and the same quantities are analogously defined 
for the second split. Lastly, we define the parameter $\lambda \in \R_{\geq 0}$ 
under each split $s \in \Set{1,2}$ to be a function that depends on the 
estimated variance of that split,
\begin{equation}
    \lambda^*_s \coloneqq \sqrt{\frac{2\logs{2/\alpha}}{m_s\hsigma_s^{*2}}} 
    \wedge \frac{1}{2c},
\label{eq:main-paper-emp-Berns-gamma-lambda-definition}
\end{equation}
for $* \in \Set{\Mirrl, \Mirru}$, and we set $c \coloneqq 1/(1-G) + 1$ which
allows both $\lambda^*_s$ and $\gamma_{E,c}(\lambda)$ to scale with respect
to the propensity score. Now that we have introduced the necessary notation, 
we proceed to present our cross-fit Studentized confidence interval.

\begin{theorem}[Cross-fit Studentized variance-adaptive confidence interval]\label{theorem:cross-fit-empirical-bernstein}
    Under either complete or Bernoulli randomization we have that
    \begin{equation}\label{eq:emp-berns-ci-l}
        L_n^{\emph{\EB}} \coloneqq \frac{1}{n} \sum_{t=1}^{\widebar T} \sum_{i \in g_t} \hite^\Mirrl_\etainvi - \Paren{\frac{\gamma_{E,c}(\lambda_2^{\Mirrl}) V_{m_1}^{\Mirrl}\brackone + \logs{2/\alpha}}{n \lambda_2^\Mirrl} + \frac{\gamma_{E,c}(\lambda_1^\Mirrl) V_{m_2}^{\Mirrl}\bracktwo + \logs{2/\alpha}}{n \lambda^{\Mirrl}_1}}
    \end{equation}
    forms a lower $(1-\alpha)$-confidence set for the average treatment effect. Under the same conditions, we let
    \begin{equation}\label{eq:emp-berns-ci-u}
        U_n^{\emph{\EB}} \coloneqq  \frac{1}{n}\sum_{t=1}^{\widebar T} \sum_{i \in g_t} \hite_\etainvi^\Mirru + 
        \frac{\gamma_{E,c}(\lambda^\Mirru_2) V_{m_1}^\Mirru\brackone + \logs{2/\alpha}}{n \lambda^\Mirru_2} + \frac{\gamma_{E,c}(\lambda^\Mirru_1) V_{m_2}^\Mirru\bracktwo + \logs{2/\alpha}}{n \lambda^\Mirru_1}
    \end{equation}
    and we have that $U_n^{\emph{\EB}}$ forms an upper $(1-\alpha)$-confidence set for the average treatment effect. Union bounding the upper and lower confidence sets we have that 
    \begin{equation}
        C_n^{\emph{\EB}} \coloneqq \Brac{L_n^{\emph{\EB}}, U_n^{\emph{\EB}}}
    \label{eq:emp-berns-ci-cross-fit}
    \end{equation}
    is a $(1-2\alpha)$-confidence interval for the average treatment effect.
\end{theorem}

\begin{figure}
    \centering
    \begin{subcaptionblock}{0.32\textwidth}
        \centering 
        \includegraphics[width=\linewidth]{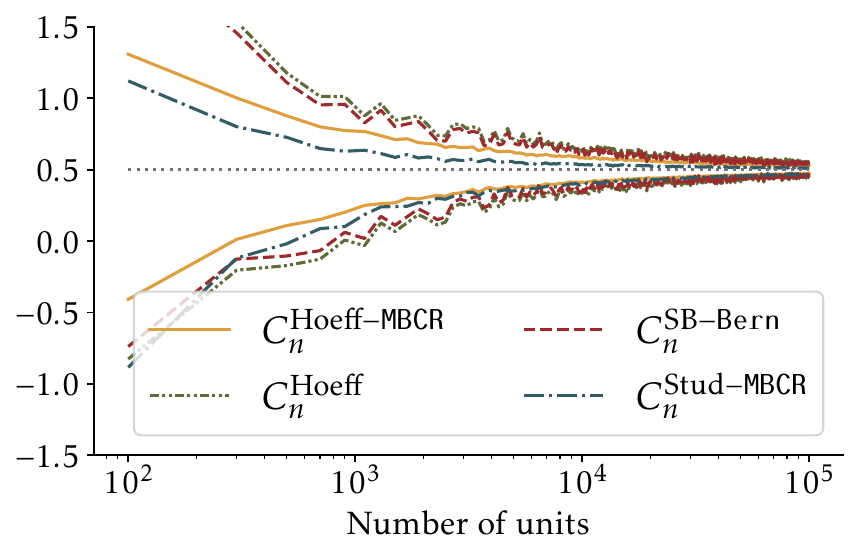}
        \caption{The potential outcomes were generated as $Y_i(0) \sim \mathrm{Unif}(0.1, 0.5)$ and $Y_i(1) \coloneqq Y_i(0) + 0.5$.}
    \end{subcaptionblock}
    \hfill
    \begin{subcaptionblock}{0.32\textwidth}
        \centering 
        \includegraphics[width=\linewidth]{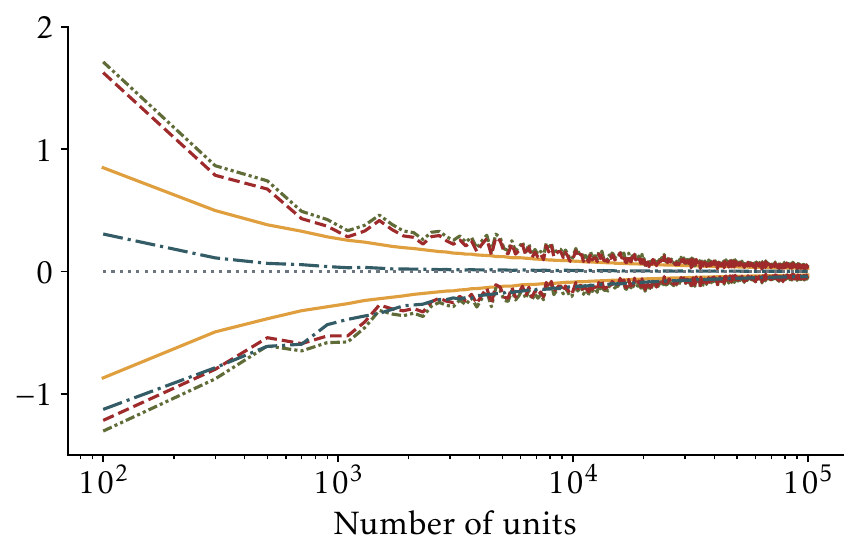}
        \caption{The potential outcomes were generated as $Y_i(0) \sim \mathrm{Unif}(0.9, 1)$ and $Y_i(1) \coloneqq Y_i(0)$.}
        \label{fig:studentized-best-performance-1}
    \end{subcaptionblock}
    \hfill
    \begin{subcaptionblock}{0.32\textwidth}
        \centering 
        \includegraphics[width=\linewidth]{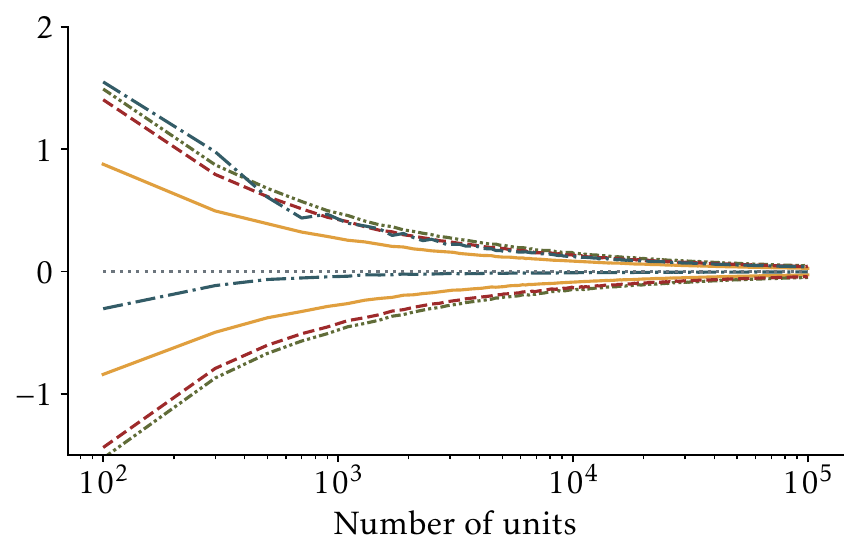}
        \caption{The potential outcomes were generated as $Y_i(0) \sim \mathrm{Unif}(0, 0.1)$ and $Y_i(1) \coloneqq Y_i(0)$.}
        \label{fig:studentized-best-performance-2}
    \end{subcaptionblock}
    \caption{Our three proposed $95\%$ nonasymptotic confidence intervals for the average treatment effect as the sample size, $n$, increases and the propensity score is set to $\pi = 0.1$. Observe that the Hoeffding-style confidence intervals under mini-batch complete randomization (Hoeff-\MBCR{}) perform reasonably well across all three scenarios, with the Studentized bounds becoming sharper in \cref{fig:studentized-best-performance-1} and \cref{fig:studentized-best-performance-2} due to the asymmetry of the potential outcomes, and hence the small empirical variance of the Horvitz-Thompson summands.
    }
    \label{fig:confidence-intervals}
\end{figure}

The proof of \cref{theorem:cross-fit-empirical-bernstein} can be found in \cref{section:proof-empirical-bernstein}. As we alluded to above, the confidence interval from \cref{theorem:cross-fit-empirical-bernstein} scales with the empirical variances, a result that we formally present in the following corollary and whose proof can also be found in \cref{section:proof-empirical-bernstein}.

\begin{corollary}
\label{corollary:eb-scaling}
    Assume $\hsigma_1^{*2} = \omega\Paren{1/n}$ and $\hsigma_2^{*2} = \omega\Paren{1/n}$ for $* \in \Set{\Mirrl, \Mirru}$ and that $m_1 = m_2 = n/2$. Then, the bounds in \eqref{eq:emp-berns-ci-l} and \eqref{eq:emp-berns-ci-u} scale for large $n$ as
    \begin{align}
    \label{eq:eb-scaling}
        \sqrt{\frac{\logs{2/\alpha}}{4 n}}
         \Paren{\frac{\widehat \sigma_1^{*2}}{\widehat \sigma^*_2} + \frac{\widehat \sigma_2^{*2}}{\widehat \sigma_1^{*}} + \widehat \sigma_2^{*} + \widehat \sigma_1^*} \quad \mathrm{for}~ * \in \Set{\Mirrl, \Mirru}\mper
    \end{align}
\end{corollary}
We are not aware of existing off-the-shelf concentration inequalities that could 
have immediately yielded a result of the kind in \cref{corollary:eb-scaling} 
under the conditions considered here. 
We can see that if the empirical variances are stationary and converging to some common $\sigma^2$, the bound scales as $2 \sigma \sqrt{\log (2/\alpha) / n}$. Therefore, the Studentized confidence interval is tightest under settings where the ``group-wise'' estimators $\ghate_t$ exhibit low variance, but can be less adequate than the Hoeffding-style and sub-Bernoulli confidence intervals in settings where the variance of the random variables is large. These behaviors can be seen in \cref{fig:confidence-intervals}, where we plot our proposed confidence intervals under different data generating processes. Importantly, notice that in settings in which the potential outcomes tend to be concentrated towards either the upper or lower bounds, our Studentized confidence interval tends to exhibit significant (one-sided) improvements over our Hoeffding-style or sub-Bernoulli confidence intervals, but the ``worse performing side'' eventually does as well as the sub-Bernoulli confidence interval.

\begin{remark}[On the form of mirrored estimators under different randomization schemes]\label{remark:the form of mirrored estimators}
The ``upper'' mirrored estimator $\frac{1}{n} \sum_{t=1}^{\widebar T} \sum_{i \in g_t} \estimator_{\etainvi}^\Mirru$ takes different forms under mini-batch complete and Bernoulli randomization. Indeed, $\estimator_\etainvi^\Mirru$ can be equivalently written as
\begin{equation}
    \estimator_{\etainvi}^\Mirru = \estimator_{\etainvi} - \Paren{\frac{Z_{\beta(i)}}{\pi} - \frac{1-Z_{\beta(i)}}{1-\pi}}  \mper
\end{equation}
Under mini-batch complete randomization (and complete randomization), exactly $n_1 = n\pi$ units are assigned to the treatment group and $n_0 = n - n_1$ are assigned to the control group, and hence
\begin{equation}
    \frac{1}{n} \sum_{t=1}^{\widebar T} \sum_{i \in g_t} \estimator_\etainvi^\Mirru = \estimator \mper
\end{equation}
Under Bernoulli randomization, however, the same estimator takes the form
\begin{equation}
    \frac{1}{n} \sum_{i=1}^n \estimator_{\etainvi}^\Mirru = \estimator - \frac{1}{n} \sum_{i=1}^n \Paren{\frac{Z_i}{\pi} - \frac{1-Z_i}{1-\pi}} \mcom
\end{equation}
with the second term being zero \emph{in expectation} but may contribute additional variance to the estimator.
\end{remark}

\section{Information-Theoretic Lower Bounds}
\label{sec:information-theoretic-bounds}
In this section we present an analysis that shows that the scaling exhibited in the confidence intervals in \cref{thm:hoeffding-mbcr,thm:sub-bernoulli-cis} is unimprovable from an information-theoretic point of view. While the appearance of confidence intervals shrinking at a rate of $1/\sqrt{n\pi}$ is not particularly surprising under $\iid$ superpopulation assumptions, but it is not as obvious that it should be optimal under the design-based setting. In this section we derive a lower bound for estimation in $L_2$ under the design-based setting which is strictly smaller than that of the $\iid$ superpopulation setting, but not in a way that changes the asymptotic scaling with respect to the effective sample size. We do so by presenting matching upper and lower bounds (up to constants) on the $L_2$ risk of estimating of the average treatment effect. 

We begin by presenting an upper bound on the estimation of the average treatment effect when using the Horvitz-Thompson estimator under mini-batch complete randomization and  under Bernoulli randomization, respectively,
\begin{equation}
\label{eq:estimation-upper-bound}
    \normtp{\hate' - \ate} \coloneqq \sqrt{\E_{\PP}\Brac{\Paren{\estimator' - \ate}^2}} \leq \frac{2}{\sqrt{n\pi}}
    \quad\mathrm{and}\quad \normtp{\hate - \ate} \leq \sqrt{\frac{2}{n\pi}}\mper
\end{equation}
The proof of the above upper bounds on the root mean square error (RMSE) can be found in \cref{section:proof-l2p-upper-bound}, and they hold under both the design-based and $\iid$ superpopulation settings. From \eqref{eq:estimation-upper-bound} we see that the root mean square error scales as $1/\sqrt{n\pi}$. This matches the scaling of the width of the Hoeffding-style and sub-Bernoulli confidence intervals presented in \cref{section:confidence-intervals}. 

In order to show the tightness of this scaling, we proceed to obtaining a corresponding lower bound on the minimax risk, defined as follows:
\begin{equation}
\label{eq:minimax-risk-def}
    \inf_{\estimator} ~ \sup_{\PP \in \cP} ~ \normtp{\estimator - \ateP}\mper
\end{equation}
The infimum in \eqref{eq:minimax-risk-def} is taken over all possible estimators, and 
we will be considering models $\cP$ for the $\iid$ superpopulation and
design-based settings. We first consider the $\iid$ superpopulation setting.

\begin{definition}[Minimax observation model for the stochastic $\iid$ setting]
\label{def:iid-stochastic-setting}
    We let $\cP_\iid$ be the set of joint distributions $\PP$ over the tuples $\Paren{\bZ, \bY}$ of treatment assignments $\bZ \coloneqq \Paren{Z_i}_{i=1}^n$ and observed outcomes $\bY \coloneqq \Paren{Y_i(Z_i)}_{i=1}^n$. The vector of observed outcomes is composed of potential outcomes tuples $\Paren{Y_i(0), Y_i(1)}_{i=1}^n$ which are $\iid$ draws from an arbitrary joint distribution with support $[0,1] \times [0,1]$. The vector of treatment assignments is assumed to be drawn marginally from some distribution $\bZ \sim \PP_{\bZ}$ induced by a design that satisfies $\Ex{n_1} = n\pi$.
\end{definition}
Note that we are not assuming that $\bZ$ are $\iid$ draws. The restriction that $\Ex{n_1} = n\pi$ clearly encapsulates both Bernoulli randomization and complete randomization (where $n_1 = n\pi$ with probability one in the latter case). Note that by virtue of there being no confounding the potential outcomes tuples 
are independent of $\bZ$ but clearly $\bZ$ and $\bY$ are highly correlated. 
Our goal, now, is to show the existence of a lower bound for 
\eqref{eq:minimax-risk-def} under the stochastic setting that scales as 
$1/\sqrt{n\pi}$. This is achieved in the following theorem.

\begin{theorem}
\label{theorem:minimax-bound-stochastic}
    The minimax risk under the root mean square error for estimating the average treatment effect under the stochastic setting (\cref{def:iid-stochastic-setting}) can be lower bounded as
    \begin{equation}
        \label{eq:minimax-bound-stochastic}
        \inf_{\estimator} \sup_{\PP \in \cP_\iid} \normtp{\estimator - \ateP} \gtrsim \frac{1}{\sqrt{n\pi}} \mper
    \label{eq:minimax-lower-bound-stochastic}
    \end{equation}
\end{theorem}
The proof---which applies to a more general setting than that considered in this section--can be found in \cref{section:proof-minimax-bound-stochastic}.  The proof is an application of Le Cam's two-point method, with some modifications to handle the dependency between the observed outcomes and the treatment assignments.

Having obtained a lower bound for \eqref{eq:minimax-risk-def} under the stochastic setting, we turn our attention to deriving a lower bound for the design-based setting, for which we consider the following observational model.

\begin{definition}[Minimax observation model for the design-based setting]
\label{def:design-based-setting}
    We let $\cP_{\DB}$ be the set of joint distributions $\PP$ over the tuples $\Paren{\bZ, \bY}$ of treatment assignments $\bZ \coloneqq \Paren{Z_i}_{i=1}^n$ and observed outcomes $\bY \coloneqq \Paren{y_i(Z_i)}_{i=1}^n$. The vector of observed outcomes is composed of potential outcome tuples $\Paren{y_i(0), y_i(1)}_{i=1}^n$ which are fixed and take values from $[0,1] \times [0,1]$. The vector of treatment assignments is assumed to be drawn marginally from some distribution $\bZ \sim \PP_{\bZ}$ induced by a design that satisfies $\Ex{n_1} = n\pi$.
\end{definition}

This model provides a design-based analog to the model from \cref{def:iid-stochastic-setting}, as it assumes the potential outcomes are fixed, deterministic values rather than being $\iid$ draws from some distribution. One would thus expect the minimax estimation problem \eqref{eq:minimax-risk-def} under the stochastic $\iid$ and design-based settings to be intrinsically related. In the following lemma we show that this is in fact the case.
\begin{lemma}
\label{lemma:minimax-relation-stochastic-design-based}
    The minimax estimation risk \eqref{eq:minimax-risk-def} under the design-based setting (\cref{def:design-based-setting}) is lower bounded by that under the stochastic i.i.d. setting (\cref{def:iid-stochastic-setting}) up to an oracle estimation term as follows
    \begin{equation}
        \inf_{\estimator}~\sup_{\PP \in \cP_\DB} \normtp{\estimator - \sate} \geq \inf_{\estimator}~\sup_{\PP \in \cP_\iid} \normtp{\estimator - \ateP} - \sup_{\PP \in \cP_\iid} \normtp{\sate - \ateP} \mcom
    \label{eq:minimax-relation-stochastic-design-based}
    \end{equation}
    where we recall that $\sate \coloneqq (1/n)\sum_{i=1}^n [y_i(1) - y_i(0)]$ and $\ateP \coloneqq \E_P\Brac{Y(1) - Y(0)}$.
\end{lemma}
\cref{eq:minimax-relation-stochastic-design-based} tells us that the minimax risk under the design-based setting is as large as the minimax risk under the stochastic $\iid$ setting plus an additional oracle estimation factor. This oracle estimation factor captures the hardness of estimating the population average treatment effect when we have oracle access to the potential outcomes.
We can combine the result from \cref{lemma:minimax-relation-stochastic-design-based} with that of \cref{theorem:minimax-bound-stochastic} and the bound from equation \eqref{eq:estimation-upper-bound} to obtain the following lower bound on the minimax risk under the design-based setting.
\begin{theorem}
\label{theorem:minimax-lower-bound-design-based}
    The minimax risk under the root mean square error for estimating the average treatment effect under the design-based setting (\cref{def:design-based-setting}) can be lower bounded as
    \begin{equation}
        \inf_{\estimator}~\sup_{\PP \in \cP_\DB} \normtp{\estimator - \sate} \gtrsim \frac{1}{\sqrt{n\pi}} - \frac{1}{\sqrt{n}} \mper
        \label{eq:minimax-lower-bound-design-based}
    \end{equation}
\end{theorem}
The proofs of \cref{lemma:minimax-relation-stochastic-design-based,theorem:minimax-lower-bound-design-based} can be found in \cref{section:proof-minimax-lower-bound-design-based}. Notice how the lower bounds from \cref{theorem:minimax-bound-stochastic,theorem:minimax-lower-bound-design-based} match the scaling of the upper bound in \eqref{eq:estimation-upper-bound}, and differ by an additive factor of $\flatfrac{1}{\sqrt{n}}$ from each other. The $1/\sqrt{n}$ corresponds to the oracle estimation factor, and it is ``negligible'' compared to the $\flatfrac{1}{\sqrt{n\pi}}$ lower bound on the stochastic minimax risk as it vanishes at a faster rate when $(n, \pi) \to (\infty, 0)$. The lower bounds given in \cref{theorem:minimax-bound-stochastic} and \cref{theorem:minimax-lower-bound-design-based} imply that estimation of the average treatment effect under the design-based setting is ``easier'' than estimation under the stochastic setting. However, the design-based setting is not \emph{fundamentally} easier than the stochastic setting because as $\pi \to 0$ the rates of \eqref{eq:minimax-lower-bound-stochastic} and \eqref{eq:minimax-lower-bound-design-based} are the same.

Taken together, the results in this section point to the optimality of the $\flatfrac{1}{\sqrt{n\pi}}$ rate observed in the Hoeffding-style and sub-Bernoulli confidence intervals from \cref{section:confidence-intervals}.

\section{Discussion}
Starting from the observation that existing nonasymptotic confidence intervals for the average treatment effect in 
randomized experiments exhibit a poorer scaling with respect to the effective sample size than asymptotic confidence intervals, we have derived nonasymptotic confidence intervals under Bernoulli randomization
and complete randomization that enjoy the same $1/\sqrt{n\pi}$ scaling as the asymptotic intervals.  We additionally  proved that this scaling is optimal by providing a matching lower bound on the statistical risk. Finally, we
also introduced a randomization procedure that is distributionally equivalent
to complete randomization, but whose properties lead to a new analysis technique that yields
Hoeffding-style and sub-Bernoulli confidence intervals that scale asymptotically
as $1/\sqrt{n\pi}$.

We note that there is related work on nonasymptotic confidence intervals for estimates of other causal effects. For instance, \citet*[Section 4.3]{tchetgen2012causal} study a setting where interference is present and they provide nonasymptotic confidence intervals for various causal estimates under a group-randomized design while assuming binary potential outcomes. Their derivation of the confidence intervals exploits the negative association of the treatment assignments in a manner similar to ours. 
An interesting future direction is, thus, to explore whether a mini-batch complete randomization-like analysis can result in tighter nonasymptotic confidence intervals under clustered randomized designs when interference is present. 

Another setting in which analogues of nonasymptotic confidence intervals have been of interest is in online, adaptive, experimentation. \citet{howard2021time} derived a confidence sequence, a sequence of confidence intervals that is uniformly valid for any sample size, for the average treatment effect. They focus on an adaptive, Bernoulli, design and provided an empirical-Bernstein confidence sequence for the average treatment effect while using the augmented inverse propensity weighting estimator. \citet{waudby2024anytime} focused on anytime-valid contextual bandit inference. As they point out, contextual bandits can be seen as a generalization of treatment effect estimation under adaptive experiments, meaning that the confidence sequences they provide in their paper also apply under the adaptive experimentation setting. As such, an additional future direction is characterizing the ``optimal'' scaling of confidence sequences for the average treatment effect, as well as extending mini-batch complete randomization to the online setting. We leave this direction for future work.

\subsection*{Acknowledgments}
The authors would like to thank David Arbour, Peng Ding, and Tobias Freidling for insightful conversations. This project was funded in part by the European Union (ERC-2022-SYG-OCEAN-101071601).
Views and opinions expressed are however those of the author(s) only and do not
necessarily reflect those of the European Union or the European Research Council
Executive Agency. Neither the European Union nor the granting authority can be
held responsible for them.  IW-S acknowledges support from the Miller Institute for Basic Research in Science.

\bibliographystyle{plainnat}
\bibliography{refs}

\appendix
\newpage
\section{Mini-batch complete randomization algorithm}
\label{sec:mbcr-algorithm}
In this section we present the formal algorithm for the mini-batch complete randomization procedure that was introduced in \cref{section:mini-batch-complete-randomization}. We start by recalling that $G \coloneqq \lceil 1/\pi \rceil$ represents the size of all but potentially the last the batch and $T$ represents the total number of groups of size $G$. Additionally, recall that $\barG$ is the size of the final group and $\barn_1$ represents the number of units assigned to the treatment arm in this final group, and that $S(N)$ is the group of permutations on $\Set{1, \dots, N}$. Lastly, recall that $a$ is the vector containing $n_0$ zeros and $n_1$ arranged in the following order: a single $1$ followed by $G-1$ zeros for $T$ times and $\barn_1$ ones followed by $\barG - \barn_1$ zeros. Having recalled the necessary notation, we now proceed to formally state the mini-batch complete randomization procedure.
\begin{algorithm}
  \caption{\textsc{Mini-batch complete randomization}~(\texttt{MBCR})}
    \label{algorithm:mbcr}
    \KwData{$n_1$ and $n_0$.}

    Compute $G, T, \barn_1$, $\barG$, and $a$.
    
    \For {$t \in \{ 1, \dots, T\}$} {
    $\beta_t \sim \mathrm{Unif}(S(G))$.
    }
    $\beta \gets \bigoplus_{t=1}^T \beta_t$
    
    \If{$\widebar G \geq 2$} {
    $\widebar \beta \sim \mathrm{Unif}(S(\widebar G))$.
    
    $\beta \gets \beta \oplus \widebar \beta$
    
    }
    
    $\eta \sim \mathrm{Unif}(S(n))$.
    
    $\rho \gets \eta \circ \beta$
    
    $(Z_1, \dots, Z_n) \gets \left ( a_{\rho(1)}, \dots, a_{\rho(n)} \right )$

    \Return{$(Z_1, \dots, Z_n)$}
\end{algorithm}

\section{Proofs of confidence intervals}
\subsection{Hoeffding-style confidence intervals}\label{section:proof-hoeffding-mbcr}

\begin{proof}[\proofref{thm:hoeffding-mbcr}]
    The proof proceeds in three steps. The first two steps are concerned with showing that a particular exponential random variable is an $e$-value---meaning that it is nonnegative and has expectation at most one. The first of the two steps exploit the between-batch conditional independence of treatment assignments as well as the marginal independence between the within-group permutations, $\beta$, and the full data $n$-permutation, $\eta$, to reduce the problem to showing that group-wise random variables are themselves $e$-values. Indeed, step 2 shows that these group-wise variables are sub-Gaussian $e$-values with sharp sub-Gaussian parameters by exploiting the \emph{dependence} between summands whose indices are permuted by $\beta$. Step 3 takes the $e$-value from the first step, applies Markov's inequality to it, optimizes the resulting confidence set using a free tuning parameter $\lambda$, and union bounds over both upper and lower bounds on the centered estimator. 

    \paragraphref{Step 1: Exploiting between-group independence to write an expectation as a product.}
    For an arbitrary $n$-permutation $\eta$, any group-wise permutation $\beta$, and any $i \in g_t$ such that $t \in [T]$, define
    \begin{equation}
      \widehat \psi_{\eta^{-1}(i)} \coloneqq \frac{Z_{\beta(i)} Y_{\eta^{-1}(i)}}{1/G} - \frac{(1-Z_{\beta(i)}) Y_{\eta^{-1}(i)}}{1-1/G},
    \end{equation}
    as well as for any $ i \in \widebar g$ as
    \begin{equation}
      \widehat \psi_{\eta^{-1}(i)} \coloneqq \frac{Z_{\widebar \beta(i)} Y_{\eta^{-1}(i)}}{1/\widetilde G} - \frac{(1-Z_{\widebar \beta(i)}) Y_{\eta^{-1}(i)}}{1-1/\widetilde G},
    \end{equation}
    where $\widetilde G = \widebar G / \widebar n_1$. We further define their centered versions as
    \begin{equation}
      \widehat \phi_{\eta^{-1}(i)} \coloneqq \widehat \psi_{\eta^{-1}(i)} - \psi_{\eta^{-1}(i)},
    \end{equation}
    where $\psi_{j} \coloneqq \E_{P}\Brac{Y_j(1) - Y_j(0)}$ is the individual treatment effect for each $ j \in [n]$. Above, we have left the dependence on $\beta$ implicit for now to reduce notational clutter, but these index permutations will appear shortly. Note that $\widehat \phi_{\eta^{-1}(i)}$ are mean-zero conditional on $\eta$.
    Consider the Hoeffding-style exponential random variable $M_n$ for any $\lambda \in \RR$ \citep{hoeffding_probability_1963},
    \begin{align}
    \label{eq:hoeffding-rv}
      M_n &\coloneqq \exp \left \{ \lambda \left ( \sum_{t=1}^{T}  \sum_{i \in g_t} \left [ \widehat \phi_{\eta^{-1}(i)} \right ]  \right ) - T \lambda^2 G^2/2 \right \}  \exp \left \{ \lambda \left (  \sum_{i \in \widebar g} \left [ \widehat \phi_{\eta^{-1}(i)} \right ]  \right ) - \lambda^2 \widebar G^2 / 2\right \} \\
          &= \exp \left \{ \lambda  \left ( \sum_{t=1}^{T} \sum_{i\in g_t} \widehat \phi_{\eta^{-1}(i)} + \sum_{i \in \widebar g} \widehat \phi_{\eta^{-1}(i)}  \right ) - \lambda^2 
            \frac{\left ( T G^2 + \widebar G^2 \right )}{2} \right \}
   \end{align}
   We will now show that $M_n$ is an $e$-value, meaning that $M_n \geq 0$ almost surely and $\Ex{M_n} \leq 1$. Indeed, nonnegativity follows by construction of taking an exponent. To show that $\EE[M_n] \leq 1$, we first condition on the random permutation $\eta$, exploit independence between $\eta$ and $\beta$, and then marginalize over the distribution of $\eta \sim \Unif(S(n))$. To that end, consider the conditional expectation of $M_n$:
   \begin{align}
     \EE \left [ M_n \mid \eta \right ] &= \EE \left [ \exp \left \{ \lambda  \left ( \sum_{t=1}^{T} \sum_{i\in g_t} \widehat \phi_{\eta^{-1}(i)} + \sum_{i \in \widebar g} \widehat \phi_{\eta^{-1}(i)}  \right ) - \lambda^2 
            \frac{\left ( T G^2 + \widebar G^2 \right )}{2} \right \} \Bigm | \eta \right ] \\
     &= \left (\prod_{t=1}^T \EE \left [ \exp \left \{ \lambda  \sum_{i\in g_t} \widehat \phi_{\eta^{-1}(i)}  - \lambda^2 
            \frac{ G^2}{2} \right \} \Bigm | \eta \right ] \right )  \EE \left [ \exp \left \{ \sum_{i \in \widebar g} \widehat \phi_{\eta^{-1}(i)} - \lambda^2  \frac{\widebar G^2}{2} \right \} \Bigm | \eta \right ],
   \end{align}
   where the second equality follows from independence of treatment permutations $\beta_1, \dots, \beta_T, \widebar \beta$ as well as independence between $\beta$ and $\eta$.

  \paragraphref{Step 2: Exploiting within-group dependence to apply Hoeffding's inequality.}
   Notice that by definition of $\beta$, we have that $\beta(i) = 1$ for one and exactly one $i \in g_t$ for each $t \in [T]$. As stated in \cref{lemma:group-wise-estimator-bound}, $\sum_{i \in g_t} \hite_{\etainvi} \in [-G, G]$; thus, implying
   \begin{equation}
     \sum_{i \in g_t} \widehat \phi_{\eta^{-1}(i)} \in \left [ - G - \sum_{i\in g_t}\psi_{\eta^{-1}(i)}, G - \sum_{i\in g_t}\psi_{\eta^{-1}(i)} \right ]\mper
   \end{equation}
   For $\widebar g$ we have that $Z_{\widebar \beta(i)} = 1$ for $\widebar n_1$ units, implying
   \begin{equation}
     \sum_{i \in \widebar g} \widehat \phi_{\eta^{-1}(i)} \in \left [ - \widebar G - \sum_{i\in g_t}\psi_{\eta^{-1}(i)}, \widebar G - \sum_{i\in g_t}\psi_{\eta^{-1}(i)} \right ].
   \end{equation}
   Consequently, $\sum_{i \in g_t}\widehat \phi_{\eta^{-1}(i)}$ is conditionally sub-Gaussian with variance proxy $(2G)^2 / 4 = G^2$ for all $t \in [T]$ and the same sum but over $ i \in \widebar g$ is conditionally sub-Gaussian but with variance proxy $\widebar G^2$, and thus
   \begin{equation}
     \EE \left [ \exp \left \{ \lambda  \sum_{i\in g_t} \widehat \phi_{\eta^{-1}(i)}  - \lambda^2 
            \frac{ G^2}{2} \right \} \Bigm | \eta \right ] \leq 1 ~~\text{and}~~\EE \left [ \exp \left \{ \lambda  \sum_{i\in \widebar g} \widehat \phi_{\eta^{-1}(i)}  - \lambda^2 
            \frac{ \widebar G^2}{2} \right \} \Bigm | \eta \right ] \leq 1,
   \end{equation}
   so that when paired with the results of Step 1,
   \begin{equation}
     \EE \left [ M_n \mid \eta \right ] \leq 1.
   \end{equation}
   Marginalizing over the distribution of $\eta$, we have that $\EE[M_n] = \EE[\EE[M_n \mid \eta]] \leq 1$.

   \paragraphref{Step 3: Inverting Markov's inequality, optimizing the width, and union bounding.}
   Applying Markov's inequality to $M_n$, we have that with probability at least $(1-\alpha)$,
    \begin{align}
      \lambda  \left ( \sum_{t=1}^{T} \sum_{i\in g_t} \left [ \widehat \psi_{\eta^{-1}(i)} - \psi_{\eta^{-1}(i)} \right ] + \sum_{i \in \widebar g} \left [ \widehat \psi_{\eta^{-1}(i)} - \psi_{\eta^{-1}(i)} \right ]  \right ) - \lambda^2 
      \frac{\left ( T G^2 + \widebar G^2 \right )}{2} < \log (1/\alpha).
    \end{align}
    Rearranging the above, we find that with probability at least $(1-\alpha)$,
    \begin{align}
      \frac{1}{n} \sum_{i=1}^{n} \left [ \widehat \psi_{\eta^{-1}(i)} - \psi_{\eta^{-1}(i)} \right ]  < \lambda 
      \frac{\left ( T G^2 + \widebar G^2 \right )}{2n} +  \frac{\log (1/\alpha)}{\lambda n}.
    \end{align}
    Finding the value of $\lambda^\star$ that minimizes the right-hand side of the above inequality, we obtain
    \begin{equation}\label{eq:bound-before-optimizing-lambda}
      \lambda^\star \coloneqq \sqrt{\frac{2\log(1/\alpha)}{T G^2 + \widebar G^2}},
    \end{equation}
    which, when plugged into \eqref{eq:bound-before-optimizing-lambda}, yields that with probability at least $(1-\alpha)$,
    \begin{align}
      \frac{1}{n} \sum_{i=1}^{n} \left [ \widehat \psi_{\eta^{-1}(i)} - \psi_{\eta^{-1}(i)} \right ]  < \sqrt{\frac{2 \log(1/\alpha)  (T G^2 + \widebar G^2)}{n^2}}.
    \end{align}
    Repeating the aforementioned one-sided bound using $-\widehat \phi_{\etainvi}$, we have that with probability at least $(1-\alpha)$,
    \begin{equation}
      \left \lvert \widehat \psi - \psi \right \rvert < \underbrace{\sqrt{\frac{2 \log(2/\alpha)  (T G^2 + \widebar G^2)}{n^2}}}_{\eqqcolon b_n(\pi)},
    \end{equation}
    which completes the proof of \cref{thm:hoeffding-mbcr}.
\end{proof}

\begin{proof}[\proofref{corollary:scaling-hoeffding-mbcr}]
    We will now analyze the width of the Hoeffding-style confidence interval given in \cref{thm:hoeffding-mbcr}, and illustrate its $\flatfrac{1}{\sqrt{n \pi}}$ scaling.
    
    Notice that $b_n(\pi)$ can be equivalently written as
    \begin{equation}\label{eq:proof-bnpi}
      b_n(\pi) \equiv \sqrt{\frac{T G^2 + \widebar G^2}{n}}  \sqrt{\frac{2\log(2/\alpha)}{n}}.
    \end{equation}
    In this step, we will show that the first factor is in some cases identically equal to $1/\sqrt{\pi}$, and in others can be upper-bounded by a quantity with the same scaling in the large-$n$, small-$\pi$ regime. We consider three exhaustive cases: $\widebar n_1 \in \{0, 1, 2\}$.

    \paragraphref{Case I: $\widebar n_1 = 0$.}
    The case where $\widebar n_1 = 0$---equivalently, when $\left \lceil 1/\pi \right \rceil = 1/\pi$---is the simplest since we have by construction that $\widebar G = n - TG = n - n_1/\pi = 0$, and hence
    \begin{align}
      b_n(\pi) &= \sqrt{\frac{TG^2}{n}}  \sqrt{\frac{2\log(2/\alpha)}{n}} \\
      \ifverbose %
      &= \sqrt{\frac{n_1 / \pi^2}{n}}  \sqrt{\frac{2\log(2/\alpha)}{n}} \\
      \fi %
      &= \sqrt{\frac{\cancel{n} \cancel{\pi} / \pi^{\cancel{2}}}{\cancel{n}}}  \sqrt{\frac{2\log(2/\alpha)}{n}} \\
      &= \frac{1}{\sqrt{\pi}}  \sqrt{\frac{2\log(2/\alpha)}{n}},
    \end{align}
    and hence we observe scaling of $1/\sqrt{n \pi}$ as alluded to before.

    \paragraphref{Case II: $\widebar n_1 = 1$.}
    In the case where $\widebar n_1 = 1$, we no longer have that $\left \lceil 1/\pi \right \rceil = 1/\pi$, and hence we will use the fact that $\left \lceil 1/\pi \right \rceil \leq 1/\pi + 1$ and justify why this conservative upper bound preserves the $1/\sqrt{\pi}$ scaling in the finite-sample, small-$\pi$ regime. 
    
    Indeed, notice that by definition of $\widebar G$ and the fact that $G = \lceil 1/\pi \rceil \geq 1/\pi$
    \begin{align}
        \widebar G &= n - TG \\
        &= n - (n_1 - \widebar n_1) G\\
        &= n - n_1G + \widebar n_1 G \\
        &\leq n - n_1 (1/\pi) + \widebar n_1 G \\
        &= \widebar n_1 G \mper
    \end{align}
    Plugging this upper bound into the first factor of \eqref{eq:proof-bnpi}, we observe that   \begin{align}
        \sqrt{\frac{TG^2 + \widebar G^2}{n}} &\leq \sqrt{\frac{TG^2 + (\widebar n_1 G)^2}{n}} \\
        &= \sqrt{\frac{\Paren{T + \widebar n_1^2} G^2}{n}}\\
        &\leq \sqrt{\frac{\Paren{T+ \widebar n_1^2}\Paren{1/\pi + 1}^2}{n}} \mcom
    \end{align}
    where in the last inequality we have used the fact that $G = \lceil 1/\pi \rceil \leq 1/\pi + 1$. Recall that $n_1 = n\pi$ and observe that by construction $T = n_1 - \barn_1$, allowing us to proceed as follows
    \begin{align}
        \sqrt{\frac{\Paren{T+ \widebar n_1^2}\Paren{1/\pi + 1}^2}{n}}
        &= \sqrt{\frac{\Paren{n_1 - \widebar n_1 + \widebar n_1^2}\Paren{1/\pi + 1}^2}{n}} \\
        &= \sqrt{\frac{n_1 \Paren{1/\pi + 1}^2 + \Paren{\widebar n_1^2 - \widebar n_1}\Paren{1/\pi + 1}^2}{n}} \\
        &= \sqrt{\frac{n\pi \Paren{\flatfrac{(1 + \pi)^2}{\pi^2}} + \Paren{\widebar n_1^2 - \widebar n_1}\Paren{1/\pi + 1}^2}{n}} \\
        &= \sqrt{\frac{\Paren{1 + \pi}^2}{\pi} + \frac{\Paren{\widebar n_1^2 - \widebar n_1}\Paren{1/\pi + 1}^2}{n}}\mper
        \label{eq:proof-bnpi-with-nbar}
    \end{align}
    Observe that when $\widebar n_1 = 1$, the second term inside the above square root is zero, allowing us to upper bound $b_n(\pi)$ as
    \begin{equation}
        b_n(\pi) \leq \sqrt{\frac{(1 + \pi)^2}{\pi}}\sqrt{\frac{2 \logs{2/\alpha}}{n}} \mcom
    \end{equation}
    paying a $(1+\pi)$ price in width over the case where $\lceil 1/\pi \rceil  = 1/\pi$, which is negligible in the small-$\pi$ regime and only $9/4$ in the worst case, when $\pi = 1/2$.

    \paragraphref{Case III: $\widebar n_1 = 2$.} In the final case where $\widebar n_1 = 2$, we can see that \eqref{eq:proof-bnpi-with-nbar} is equal to
    \begin{equation}
        \sqrt{\frac{\Paren{1 + \pi}^2}{\pi} + \frac{2\Paren{1/\pi + 1}^2}{n}}\mcom
    \end{equation}
    which has the same scale as before in the large-$n$ regime. This completes the proof.
\end{proof}

\subsection{Sub-Bernoulli confidence sets for the Horvitz-Thompson estimator}\label{section:proof-sub-bernoulli-ci}
We will handle the proofs of \cref{thm:sub-bernoulli-cis} under Bernoulli and mini-batch complete randomization simultaneously. Recall that under mini-batch complete randomization we work with the following slightly modified estimator for the individual treatment effects
\begin{equation}
    \estimator_{\etainvi} \coloneqq Y_{\etainvi}\Paren{\frac{Z_{\beta(i)}}{1/G} - \frac{1-Z_{\beta(i)}}{1 - 1/G}} \mcom 
\end{equation}
where both $\eta^{-1}(\cdot)$ and $\beta(\cdot)$ are identity maps under Bernoulli randomization, and they are the unit-wide and withing sub-group random permutations under mini-batch complete randomization. We use the above random variables for units belonging to the groups $g_t$ for $t \in [T]$, where under Bernoulli randomization we let $g_t = \Set{t}$ for $t \in [n]$ and under mini-batch complete randomization $g_t$ is the set of units belonging to group $t \in [T]$. Moreover, recall that
under mini-batch complete randomization we might have a final group $\widebar{g}$. To simplify the notation in this group we define $\widebar{T} \coloneqq T + \Ind{\barG > 0}$ and as a convention throughout this proof we use $g_{\widebar{T}}$ to denote group $\widebar{g}$ whenever $\Ind{\barG > 0}$. Under Bernoulli randomization we let $\widebar{T} = n$. Moreover, the estimators for the individual treatment effects for units in group $g_{\widebar{T}}$ is given by 
\begin{equation}
  \widehat \psi_{\eta^{-1}(i)} \coloneqq \frac{Z_{\widebar \beta(i)} Y_{\eta^{-1}(i)}}{1/\widetilde G} - \frac{(1-Z_{\widebar \beta(i)}) Y_{\eta^{-1}(i)}}{1-1/\widetilde G},
\end{equation}
where we recall that $\widetilde G = \widebar G / \widebar n_1$. 

Throughout the proof we will be working with the following centered random variable
\begin{equation}
    \phi_t \coloneqq \ghate_t - \vartheta_t \mcom \quad\mathrm{where}\quad \ghate_t \coloneqq \sum_{i \in g_t} \estimator_{\etainvi} \quad\mathrm{and}\quad \vartheta_t \coloneqq \Ex{\ghate_t}
\end{equation}
for each $t \in [\widebar{T}]$.
For most of the following proof, we will be working with the centered random variable $\phi_t$. Because in this work we are under a setting in which the potential outcomes take values within the unit interval (Assumption~\ref{assumption:generalized-potential-outcomes}), the centered random variable takes values $\phi_t \in [-2G, 2G]$ for $t \in [T]$ and takes values $\phi_{\widebar{T}} \in [-2 \widebar{G}, 2\widebar{G}]$ for the last group under mini-batch complete randomization. Under Bernoulli randomization the centered random variable takes values $\phi_t \in [-1/(1-\pi) - 1, 1/\pi + 1 ]$ for each $t \in [\widebar{T}]$. To unify the notation, in the following proof we will be denoting the lower bound of $\phi_t$ as $a_t$ and its upper bound as $b_t$ so that $\phi_t \in [a_t, b_t]$ for each $t \in [T]$.

\begin{proof}[\proofref{thm:sub-bernoulli-cis}]
    The proof proceeds in three steps. In the first step we show that a particular exponential random variable is an $e$-value. In the second step we apply Markov's inequality and union bound over both the upper and lower bounds of the centered estimator. In the third step we write the resulting bound from step two in terms of the specific quantities under Bernoulli and mini-batch complete randomization.

    \paragraphref{Step 1: Showing that an exponential random variable is an $e$-value.} We start by considering the following exponential random variable $M_T$ for any $\lambda \in \R$:
    \begin{equation}
        M_T \coloneqq \exps{\sum_{t=1}^{\widebar{T}} \Paren{\lambda \phi_t - \gamma_{B,t}(\lambda)}} \mcom
    \label{eq:sub-bernoulli-rv}
    \end{equation} 
    where we have defined the function
    \begin{equation}
        \gamma_{B, t}(\lambda) \coloneqq \logs{\frac{b_t}{b_t - a_t}e^{\lambda a_t} - \frac{a_t}{b_t - a_t}e^{\lambda b_t}} \mper
    \end{equation}

    Recall that our current goal is to show that the \eqref{eq:sub-bernoulli-rv} is an $e$-value. By construction we know that $M_{\widebar{T}} \geq 0$ almost surely; so, all it remains to do is to show that $\Ex{M_{\widebar{T}}} \leq 1$. By the law of total expectation we have that $\Ex{M_{\widebar{T}}} = \Ex{\Ex{M_{\widebar{T}} ~\vert~ \eta}}$. So for now we will focus on upper bounding the conditional expectation. Indeed, due to the independence between the groups when we condition on $\eta$ we can rewrite the conditional expectation as
    \begin{equation}
        \Ex{M_{\widebar{T}} ~\vert~ \eta} = \Ex{\prod_{t=1}^{\widebar{T}} \exps{\lambda \phi_t - \gamma_{B,t}(\lambda)} \svert \eta} = \prod_{t=1}^{\widebar{T}} \Ex{\exps{\lambda \phi_t - \gamma_{B,t}(\lambda)} ~\vert~ \eta} \mper
        \label{eq:conditional-expectation-sub-bernoulli-rv}
    \end{equation}
    We appeal to \citet[Lemma 1]{hoeffding_probability_1963} which states that
    \begin{equation}
        \Ex{\exps{\lambda \phi_t}} \leq \frac{b_t}{b_t - a_t} e^{\lambda a_t} - \frac{a_t}{b_t - a_t} e^{\lambda b_t}
        \label{eq:sub-bernoulli-hoeffding-bound}
    \end{equation}
    for any $\lambda \in \R$ and $t = 1, \dots, \widebar{T}$, and we can see how 
    \eqref{eq:sub-bernoulli-hoeffding-bound} implies
    \begin{equation}
        \Ex{\exps{\lambda \phi_t - \gamma_{B,t}(\lambda)} ~\vert~ \eta} \leq 1 \mper
        \label{eq:sub-bernoulli-boundary}
    \end{equation}
    Hence, putting \eqref{eq:conditional-expectation-sub-bernoulli-rv} and \eqref{eq:sub-bernoulli-boundary} together we can conclude that
    \begin{equation}
        \Ex{M_{\widebar{T}} ~\vert~ \eta} \leq 1 \mper
    \end{equation}
    Moreover, due to the law of total expectation we can, further, conclude that $\Ex{M_{\widebar{T}}} \leq 1$, showing that $M_{\widebar{T}}$ is an $e$-value.

    \paragraphref{Step 2: Inverting Markov's inequality and union bounding.}
    We now proceed to apply Markov's inequality to $M_{\widebar{T}}$, from which we can see that
    \begin{equation}
        \PP\Paren{M_{\widebar{T}} \geq 1/\alpha} \leq \frac{\Ex{M_{\widebar{T}}}}{1/\alpha} \leq \alpha \mper
    \end{equation}
    Thus, we have that with probability at least $(1-\alpha)$,
    \begin{equation}
        \exps{\sum_{t=1}^{\widebar{T}} \Paren{\lambda \phi_t - \gamma_{B,t}(\lambda)}} \leq 1/\alpha \mper
    \end{equation}
    Rearranging the above terms, we find that with probability at least $(1-\alpha)$,
    \begin{equation}
        \frac{1}{n} \sum_{t=1}^{\widebar{T}} \ghate_t - \vartheta_t \leq \frac{\logs{1/\alpha} + \sum_{t=1}^{\widebar{T}} \gamma_{B, t}(\lambda)}{\lambda n} \mper
    \end{equation}
    Applying the same bound to $-\phi_1, \dots, -\phi_{\widebar{T}}$ and union bounding we have that with probability at least $(1-\alpha)$,
    \begin{equation}
        \Abs{\frac{1}{n} \sum_{t=1}^{\widebar{T}} \ghate_t - \vartheta_t} \leq \frac{\logs{2/\alpha} + \sum_{t=1}^{\widebar{T}} \gamma_{B, t}(\lambda)}{\lambda n} \mper
        \label{eq:proof-sub-bernoulli-ci-bound}
    \end{equation}

    \paragraphref{Step 3: Specializing the confidence interval to the respective randomization procedure.}
    Notice how the bound from \eqref{eq:proof-sub-bernoulli-ci-bound} is written in terms of $a_t$ and $b_t$ for $t \in [\widebar{T}]$, and that it is written in terms of a unified notation that allowed us to prove the result under mini-batch and Bernoulli randomization concurrently. We now proceed to specialize the confidence interval from \eqref{eq:proof-sub-bernoulli-ci-bound} for mini-batch complete randomization and Bernoulli randomization.
    
    \emph{Confidence interval under mini-batch complete randomization:} We first focus on the mini-batch complete randomization setting, in which the bounds of the centered random variables $\phi_t$ take the values $a_t \coloneqq -2G$ and $b_t \coloneqq 2G$ for each $t \in [T]$ and $a_{\widebar{T}} \coloneqq -2\widebar{G}$ and $b_{\widebar{T}} \coloneqq 2 \widebar{G}$ for the last group $g_{\widebar{T}}$. We now proceed to write $\gamma_{B,t}(\lambda)$ in terms of the bounds under mini-batch complete randomization. As such, the cumulant generating function for $t \in [T]$ becomes
    \begin{align}
        \gamma_{B, t}(\lambda) &= \logs{\frac{b}{b-a} e^{\lambda a} - \frac{a}{b-a} e^{\lambda b}} \\
        &= \logs{\frac{2G}{4G} e^{-2G \lambda} + \frac{2G}{4G} e^{2G \lambda}} \\
        &= \logs{\frac{1}{2} e^{-2G \lambda} + \frac{1}{2} e^{2G \lambda}} \mper
    \end{align}
    Following the same steps we can see how the cumulant generating for group $g_{\widebar{T}}$ is
    \begin{equation}
        \gamma_{B, \widebar{T}} = \logs{\frac{1}{2}e^{-2\widebar{G}\lambda} + \frac{1}{2}e^{2\widebar{G}\lambda}} \mper
    \end{equation}
    Substituting the above cumulant generating functions to the boundary of \eqref{eq:proof-sub-bernoulli-ci-bound} gives us that with probability at least $(1-\alpha)$,
    \begin{equation}
        \Abs{\frac{1}{n} \sum_{t=1}^{\widebar{T}} \ghate_t - \vartheta_t} \leq 
        \frac{\logs{2/\alpha} + T \logs{\flatfrac{e^{-2G \lambda}}{2}  + \flatfrac{e^{2G \lambda}}{2}} + \logs{e^{-2\widebar{G}\lambda}/2 + e^{2\widebar{G}\lambda}/2}}{\lambda n} \mcom
    \label{eq:proof-sub-bernoulli-ci-mbcr}
    \end{equation}
    for any $\lambda \in \R$.

    \emph{Boundary under Bernoulli randomization:}
    Under Bernoulli randomization the bounds of the centered random variables take the values $a_t \coloneqq -1/(1-\pi) - 1$ and $b_t \coloneqq 1/\pi + 1$ for each $t \in [n]$. Hence, we will drop the dependency on $t$ in $\gamma_{B,t}(\lambda)$ and we will write the cumulant generating function as
    \begin{equation}
        \gamma_{B}(\lambda) \coloneqq \logs{\frac{b}{b-a}e^{\lambda a} - \frac{a}{b-a}e^{\lambda b}}
    \end{equation}
    for $a \coloneqq -1/(1-\pi) - 1$ and $b \coloneqq 1/\pi + 1$.
    So, we can immediately see how \eqref{eq:proof-sub-bernoulli-ci-bound} implies that with probability $(1-\alpha)$,
    \begin{equation}
        \Abs{\frac{1}{n} \sum_{i=1}^n \hite_i - \ite_i} \leq \frac{\logs{2/\alpha} + n \gamma_{B}(\lambda)}{\lambda n} \mper
    \label{eq:proof-sub-bernoulli-ci-bernoulli}
    \end{equation}
\end{proof}

\begin{proof}[Scaling of the sub-Bernoulli confidence intervals]
    We will now show that the sub-Bernoulli confidence intervals from \eqref{eq:proof-sub-bernoulli-ci-mbcr} and \eqref{eq:proof-sub-bernoulli-ci-bernoulli} scale as $\Bigoh(1/\sqrt{n\pi})$.
    
    \paragraphref{Scaling under mini-batch complete randomization.}
    
    We will start by showing that for any $\lambda \in \R$,
    \begin{equation}
        \frac{\logs{\frac{1}{2} e^{-2G \lambda} + \frac{1}{2} e^{2G \lambda}} + \logs{\frac{1}{2} e^{-2\widebar{G} \lambda} + \frac{1}{2} e^{2\widebar{G} \lambda}}}{\lambda^2 / 2} \to 4G^2 + 4\widebar{G}^2 \quad\mathrm{as}\quad \lambda \to 0 \mper
        \label{eq:proof-scaling-of-sub-bernoulli-mbcr-initial-scaling}
    \end{equation}
    By one application of L'H\^opital's rule we can reduce the limit from \eqref{eq:proof-scaling-of-sub-bernoulli-mbcr-initial-scaling} to one which is given by
    \begin{equation}
        \frac{\Paren{-Ge^{-2G \lambda} + G e^{2G\lambda}} \big/ \Paren{e^{-2G\lambda} / 2 + e^{2G\lambda}/2} + \Paren{-\widebar{G}e^{-2\widebar{G} \lambda} + \widebar{G} e^{2\widebar{G} \lambda}} \big/ \Paren{e^{-2\widebar{G} \lambda} / 2 + e^{2\widebar{G} \lambda}/2}}{\lambda} \mper
    \end{equation}
    Through another application of L'H\^opital's rule, taking $\lambda \to 0$, and simplifying the resulting expression we can see how
    \begin{equation}
        \frac{\logs{\frac{1}{2} e^{-2G \lambda} + \frac{1}{2} e^{2G \lambda}} + \logs{\frac{1}{2} e^{-2\widebar{G} \lambda} + \frac{1}{2} e^{2\widebar{G} \lambda}}}{\lambda^2 / 2} \to 4G^2 + 4\widebar{G}^2 \mper
    \end{equation}

    Throughout the rest of the analysis of the scaling of the sub-Bernoulli confidence interval under mini-batch complete randomization we will let
    \begin{equation}
        \lambda_{\MBCR}  \coloneqq \sqrt{\frac{2\logs{2/\alpha}}{4TG^2 + 4\widebar{G}^2}} \mcom
    \end{equation}
    but for notational simplicity we will work with $\lambda \equiv \lambda_{\MBCR}$ for the remainder of the scaling analysis.
    Using the above limit and the definition of $\lambda_{\MBCR}$ we proceed to analyze the scaling of \eqref{eq:proof-sub-bernoulli-ci-mbcr} as follows
    \begin{align}
       &\frac{\logs{2/\alpha} + T \logs{\flatfrac{e^{-2G \lambda}}{2}  + \flatfrac{e^{2G \lambda}}{2}} + \logs{\flatfrac{e^{-2\widebar{G} \lambda}}{2}  + \flatfrac{e^{2\widebar{G} \lambda}}{2}}}{\lambda n}\\ 
       &\quad= \frac{\logs{2/\alpha} + (\lambda^2/2)\Brac{T\logs{\flatfrac{e^{-2G \lambda}}{2}  + \flatfrac{e^{2G \lambda}}{2}} + \logs{\flatfrac{e^{-2\widebar{G} \lambda}}{2}  + \flatfrac{e^{2\widebar{G} \lambda}}{2}}} \big/ (\lambda^2/2)}{\lambda n} \\
       &\quad\asymp \frac{\logs{2/\alpha} + (\lambda^2 / 2) \Brac{T 4G^2 + 4\widebar{G}^2}}{\lambda_ n} \\
       &\quad= \frac{\logs{2/\alpha} + \logs{2/\alpha}}{\lambda n}
       \label{eq:proof-scaling-of-sub-bernoulli-mbcr-first-substitution-of-lambda} \\
       &\quad= \frac{2 \logs{2/\alpha}}{n \sqrt{\Paren{2\logs{2/\alpha}} \big/ \Paren{4TG^2 + 4\widebar{G}^2}}}
       \label{eq:proof-scaling-of-sub-bernoulli-mbcr-second-substitution-of-lambda} \\
       &\quad= \frac{\sqrt{8 \logs{2/\alpha} \Paren{TG^2 + \widebar{G}^2}}}{n}\mcom
       \label{eq:proof-scaling-of-sub-bernoulli-mbcr-intermediate-scaling}
    \end{align}
    where \eqref{eq:proof-scaling-of-sub-bernoulli-mbcr-first-substitution-of-lambda} and \eqref{eq:proof-scaling-of-sub-bernoulli-mbcr-second-substitution-of-lambda} follow from substituting in the value of $\lambda_{\MBCR}$ and simplifying the expression.

    \emph{Case I: $\barn_1 = 0$.}\quad
    For the scenario in which $\barn_1 = 0$---meaning that $\pi = 1/K$ for an integer $K \geq 2$---we have that $TG = n$ and that $\widebar{G} = 0$. In this scenario \eqref{eq:proof-scaling-of-sub-bernoulli-mbcr-intermediate-scaling} can be further simplified to
    \begin{equation}
       \eqref{eq:proof-scaling-of-sub-bernoulli-mbcr-intermediate-scaling}
       = \sqrt{\frac{8\logs{2/\alpha}}{n\pi}} \mcom
       \label{eq:proof-scaling-of-sub-bernoulli-mbcr-value-of-G}
    \end{equation}

    \emph{Case II: $\barn_1 = 1$.}\quad
    For this scenario we follow the same steps as in the proof of \cref{corollary:scaling-hoeffding-mbcr} found in \cref{section:proof-hoeffding-mbcr} to upper bound $\sqrt{TG^2 + \widebar{G}^2}$ as
    \begin{equation}
        \sqrt{TG^2 + \widebar{G}^2} \leq \sqrt{n ((1 + \pi)^2 / \pi) + (\barn_1^2 - \barn_1)(1/\pi + 1)^2} \mper
    \label{eq:proof-scaling-of-sub-bernoulli-mbcr-bound-on-scaling}
    \end{equation}
    Hence, for the case in which $\barn_1 = 1$ we have that
    \begin{equation}
        \frac{\sqrt{8 \logs{2/\alpha} \Paren{TG^2 + \widebar{G}^2}}}{n} \leq \sqrt{\frac{8\logs{2/\alpha} (1 + \pi)^2}{n\pi}} \mper
    \end{equation}

    \emph{Case III: $\barn_1 = 2$.}\quad
    We use similar arguments as in Case II and leverage the upper bound from \eqref{eq:proof-scaling-of-sub-bernoulli-mbcr-bound-on-scaling}. Nonetheless, for the scenario in which $\barn_1 = 2$ we have that
    \begin{align}
        \frac{\sqrt{8 \logs{2/\alpha} \Paren{TG^2 + \widebar{G}^2}}}{n} &\leq
        \frac{\sqrt{8\logs{2/\alpha} \Brac{n((1 + \pi)^2 / \pi) + 2(1/\pi)^2}}}{n} \\
        &\leq \sqrt{\frac{8 \logs{2/\alpha}(1+\pi)^2}{n\pi}} + \frac{\sqrt{16\logs{2/\alpha}(1/\pi)^2}}{n}
        \mcom
    \end{align}
    where the last inequality follows from the subadditivity of the square root function.

    We can, thus, see that in all three cases \eqref{eq:proof-sub-bernoulli-ci-mbcr}
    scales as $\Bigoh\Paren{1/\sqrt{n\pi}}$.

    \paragraphref{Scaling under Bernoulli randomization.}
    We will start by showing that for any $\lambda  \in \R$,
    \begin{equation}
        \frac{\gamma_{B}(\lambda)}{\lambda^2 / 2} \to \Paren{\frac{1}{1-\pi} + 1} \Paren{\frac{1}{\pi} + 1} \quad\mathrm{as}\quad \lambda \to 0 \mper
        \label{eq:proof-scaling-sub-bernoulli-bernoulli-initial-scaling}
    \end{equation}
    Indeed, by one application of L'H\^opital's rule we can reduce the limit from \eqref{eq:proof-scaling-sub-bernoulli-bernoulli-initial-scaling} to one which is given by
    \begin{equation}
        \lim_{\lambda \to 0}\frac{\frac{ab}{b-a}\Paren{e^{\lambda a} - e^{\lambda b}} \big/\Paren{\frac{b}{b-a}e^{\lambda a} - \frac{a}{b-a}e^{\lambda b}}}{\lambda} \mper
    \label{eq:proof-scaling-sub-bernoulli-bernoulli-first-lhopitals}
    \end{equation}
    Through another application of L'H\^opital's rule and simplifying the resulting expression we can see from
    \eqref{eq:proof-scaling-sub-bernoulli-bernoulli-first-lhopitals} that
    \begin{equation}
        \lim_{\lambda \to 0}\frac{\gamma_{B}(\lambda)}{\lambda^2 / 2} = \Paren{\frac{1}{1-\pi} + 1} \Paren{\frac{1}{\pi} + 1} \mper
    \end{equation}

    Throughout the rest of the analysis of the scaling of the sub-Bernoulli confidence interval under Bernoulli randomization we will let
    \begin{equation}
        \lambda_{\tBern} \coloneqq \sqrt{\frac{2 \logs{2/\alpha}}{n \Paren{1/[1-\pi] + 1}\Paren{1/\pi + 1}}} \mcom
    \end{equation}
    but to simplify the notation in the remainder of the analysis we will let $\lambda \equiv \lambda_{\tBern}$
    Using the above limit and the definition of $\lambda_{\Bern}$ we proceed to analyze the scaling of \eqref{eq:proof-sub-bernoulli-ci-bernoulli} as follows
    \begin{align}
        \frac{\logs{2/\alpha} + n \gamma_{B}(\lambda)}{\lambda_ n} &= \frac{\logs{2/\alpha} + (\lambda^2 / 2)n\gamma_{B}(\lambda) \big/ (\lambda^2 / 2)}{\lambda n} \\
        &\asymp \frac{\logs{2/\alpha} + \Paren{\lambda^2 / 2} n \Paren{1/[1-\pi] + 1}\Paren{1/\pi + 1}}{\lambda n} \\
        &= \frac{\logs{2/\alpha} + \logs{2/\alpha}}{\lambda n} 
        \label{eq:proof-scaling-sub-bernoulli-scaling-lambda_substitution-one}\\
        &= \sqrt{\frac{2\logs{2/\alpha} \Paren{\frac{1}{1-\pi} + 1}\Paren{\frac{1}{\pi} + 1}}{n}}
        \label{eq:proof-scaling-sub-bernoulli-scaling-lambda_substitution-two}\\
    \end{align}
    where \eqref{eq:proof-scaling-sub-bernoulli-scaling-lambda_substitution-one} and \eqref{eq:proof-scaling-sub-bernoulli-scaling-lambda_substitution-two} follow from substituting in the value of $\lambda_{\Bern}$ and simplifying the expression. We  
    can in turn see how in the large-$n$, small-$\pi$ regime \eqref{eq:proof-sub-bernoulli-ci-bernoulli} scales as
    \begin{equation}
        \sqrt{\frac{4 \logs{2/\alpha}}{n\pi}} \mper
    \end{equation}
    Nonetheless, we can further upper bound \eqref{eq:proof-scaling-sub-bernoulli-scaling-lambda_substitution-two} as follows
    \begin{align}
    \eqref{eq:proof-scaling-sub-bernoulli-scaling-lambda_substitution-two}
        &\leq \sqrt{\frac{6 \logs{2/\alpha} \Paren{\frac{1}{\pi} + 1}}{n}} \label{eq:proof-scaling-sub-bernoulli-scaling-one-minus-pi-upper-bound}\\
        &= \sqrt{\frac{6\logs{2/\alpha}}{n\pi} + \frac{6 \logs{2/\alpha}}{n}} \mcom
    \end{align}
    where \eqref{eq:proof-scaling-sub-bernoulli-scaling-one-minus-pi-upper-bound} follows from the upper bound $1/(1-\pi) \leq 2$. We, thus, have that \eqref{eq:proof-sub-bernoulli-ci-bernoulli} scales as $\Bigoh \Paren{1/\sqrt{n\pi}}$, completing our proof.
\end{proof}

\subsection{Studentized variance-adaptive confidence interval}\label{section:proof-empirical-bernstein}
In this subsection we provide the proof of the empirical-Bernstein confidence interval from \cref{theorem:cross-fit-empirical-bernstein}. Hence, throughout we will be working with the following ``mirrored'' estimators
\begin{equation}
\label{eq:proof-mirrored-estimators}
    \hite_\etainvi^\Mirrl \coloneqq Y_{\etainvi}  \Paren{\frac{Z_\betai}{\flatfrac{1}{G}} - \frac{1 - Z_\betai}{1 - \flatfrac{1}{G}}} \quad\text{and}\quad
    \hite_\etainvi^\Mirru \coloneqq (Y_\etainvi - 1)  \Paren{\frac{Z_\betai}{\flatfrac{1}{G}} - \frac{1 - Z_\betai}{1 - \flatfrac{1}{G}}} \mcom
\end{equation}
used in the lower and upper confidence sets, respectively, for units who belong to groups $g_t$ for $t = 1, \dots, T$. We note that while the two estimators have different forms, their means are identical since $Z_{\beta(i)} / (1/G) - (1-Z_{\beta(i)}) / (1-1/G)$ has mean zero conditional on $\eta$. Meanwhile, the ``mirrored'' estimators we use for units $i$ who belong to group $\widebar g$ are
\begin{equation}
\label{eq:proof-mirrored-estimators-extra-group}
    \hite_\etainvi^\Mirrl \coloneqq Y_{\etainvi}  \Paren{\frac{Z_\betai}{\flatfrac{1}{\widetilde G}} - \frac{1 - Z_\betai}{1 - \flatfrac{1}{\widetilde G}}} \quad\text{and}\quad
    \hite_\etainvi^\Mirru \coloneqq (Y_\etainvi - 1)  \Paren{\frac{Z_\betai}{\flatfrac{1}{\widetilde G}} - \frac{1 - Z_\betai}{1 - \flatfrac{1}{\widetilde G}}} \mper
\end{equation}
We will be simultaneously handling the cases in which the treatment assignments are determined through Bernoulli or mini-batch complete randomization. Both of the maps $\eta^{-1}(\cdot)$ and $\beta(\cdot)$ are identity maps under Bernoulli randomization, and they are the unit-wide and withing sub-group random permutations under mini-batch complete randomization.

Throughout the proof, we will be working with the following random variables
\begin{equation}
    \ghate^*_t \coloneqq \sum_{i \in g_t} \hite_\etainvi^* \quad\text{for}~* \in \Set{\Mirrl, \Mirru} \mcom
\end{equation}
where under Bernoulli randomization $g_t = \Set{t}$ for $t \in [n]$, and under mini-batch complete randomization $g_t$ it is the set of units belonging to group $t$ for $t = 1, \dots, \widebar T$, where we define $\widebar T = T + \Ind{\widebar G > 0}$ and as a convention throughout this proof we will use $g_{\widebar T}$ to denote the group $\widebar g$ whenever $\Ind{\widebar G > 0} = 1$.

\begin{proof}[\proofref{theorem:cross-fit-empirical-bernstein}]
    This proof will proceed in three steps. First, we will show that the random variable~\eqref{eq:proof:emp-berns-cross-fit-rv} forms an $e$-value. We will then apply Markov's inequality to obtain a high-probability upper-bound on the the difference between the Horvitz-Thompson estimator and the average treatment effect, giving us a lower confidence set for the average treatment effect. We will then perform the previous two steps on a ``mirrored'' estimator in order to derive a high-probability upper-bound, which will then yield an upper confidence set for the same effect. Finally, we will union bound over these two confidence sets to arrive at the desired confidence interval.
    
    Throughout this proof we will be working with two quantities, $\cD_1 \coloneqq \Set{\ghate^{\Mirrl}_1, \dots, \ghate^{\Mirrl}_{m_1}}$ and $\cD_2 \coloneqq \Set{\ghate^{\Mirrl}_{m_1 + 1}, \dots, \ghate^{\Mirrl}_{\widebar T}}$ which denote the sample splits, and where we define $m_1 \coloneqq \lfloor \flatfrac{\widebar T}{2} \rfloor$ and $m_2 \coloneqq \widebar T - m_1$. For mini-batch complete randomization we will view $\eta$ as the permutation that induces the ordering of the groups $g_t$, and in the case of Bernoulli randomization we will think of $\eta$ as inducing the ordering of the units, and, thus, the sample splits. 
    Moreover, throughout this proof we will be working with the following exponential random variable
    \begin{equation}
    \label{eq:proof:emp-berns-cross-fit-rv}
        E_{m_1}^{\Mirrl}\brackone \coloneqq \exps{\lambda^{\Mirrl}_2 \sum_{t=1}^{m_1}\Paren{\ghate^{\Mirrl}_t - \vartheta_t} - \gamma_{E,c}(\lambda^{\Mirrl}_2)  V_{m_1}^{\Mirrl}\brackone} \mcom
    \end{equation}
    where $V_{m_1}^{\Mirrl}\brackone$, $\lambda^{\Mirrl}_2$, and $\gamma_{E,c}(\lambda) : (0, 1/c) \to [0, \infty)$ are defined as in \eqref{eq:main-paper-emp-Bernstein-avg-variance-definition} and \eqref{eq:main-paper-emp-Berns-gamma-lambda-definition}. The function $\gamma_{E,c}(\lambda)$ depends on the scale parameters $\lambda \in \R_{\geq 0}$ and $c \in \R_{> 0}$. We let $c = [1/(1-G)] \vee [1/(1-\widetilde{G})] + 1$ as these are the (lower and upper) bounds of the centered random variables we work with in the construction of the Studentized confidence interval.
    
    \paragraphref{Step 1: Showing that $E_{m_1}^{\Mirrl}\brackone$ is a conditional $e$-value.}
    We will first show that $E_{m_1}^{\Mirrl}\brackone$ is a conditional $e$-value given the permutation $\eta$ and the second split of the data $\cD_2$. To do so, we will apply Fan's inequality to~\eqref{eq:proof:emp-berns-cross-fit-rv}. Before proceeding, however, we note that we can equivalently write~\eqref{eq:proof:emp-berns-cross-fit-rv} as
    \begin{equation}
        E_{m_1}^{\Mirrl}\brackone = \prod_{t=1}^{m_1} \exps{\lambda^{\Mirrl}_2 \Paren{\ghate_t^{\Mirrl} - \vartheta_t} - \gamma_{E,c}(\lambda^{\Mirrl}_2)  \Paren{\ghate^{\Mirrl}_t - \hmu_{t-1}^{\Mirrl}\brackone}^2} \mper
    \end{equation}
    Furthermore, we will denote as $\bghate^{\Mirrl}_{s:t} \coloneqq \Paren{\ghate^{\Mirrl}_s, \dots, \ghate^{\Mirrl}_t}$ the collection of group-wise Horvitz-Thompson estimators starting with group $s$ and going up to group $t$.
    
    We now proceed by taking the expectation of~\eqref{eq:proof:emp-berns-cross-fit-rv} and using the tower property of conditional expectations
    \begin{align}
        \Ex{E_{m_1}^{\Mirrl}\brackone} &= \E_{\cD_2, \eta}\Brac{\Ex{\prod_{t=1}^{m_1}\exps{\lambda^{\Mirrl}_2 \Paren{\ghate^{\Mirrl}_t - \vartheta_t} - \gamma_{E,c}(\lambda^{\Mirrl}_2)  \Paren{\ghate^{\Mirrl}_t - \hmu_{t-1}^{\Mirrl}\brackone}^2}\svert \cD_2, \eta}}\\
        &= \E_{\cD_2, \eta}\Brac{\E_{\bghate^{\Mirrl}_{1:m_1-1}}\Brac{\Ex{\prod_{t=1}^{m_1} \exps{\lambda^{\Mirrl}_2 \Paren{\ghate^{\Mirrl}_t - \vartheta_t} - \gamma_{E,c}(\lambda^{\Mirrl}_2)  \Paren{\ghate_t^{\Mirrl} - \hmu_{t-1}^{\Mirrl}\brackone}^2} \svert \bghate^{\Mirrl}_{1:m_1 - 1}, \cD_2, \eta} \svert \cD_2, \eta}} \\
        &= \E_{\cD_2, \eta}\left[\E_{\bghate^{\Mirrl}_{1:m_1 - 1}}\left[\prod_{t=1}^{m_1 - 1} \exps{\lambda^{\Mirrl}_2 \Paren{\ghate^{\Mirrl}_t - \vartheta_t} - \gamma_{E,c}(\lambda^{\Mirrl}_2)  \Paren{\ghate^{\Mirrl}_t - \hmu_{t-1}^{\Mirrl}\brackone}^2}\right.\right. \\
        &\qquad \left.\left.  \Ex{\exps{\lambda^{\Mirrl}_2\Paren{\ghate^{\Mirrl}_{m_1} - \vartheta_{m_1}} - \gamma_{E,c}(\lambda^{\Mirrl}_2)  \Paren{\ghate^{\Mirrl}_{m_1} - \hmu_{m_1 - 1}^{\Mirrl}\brackone}^2} \svert \bghate^{\Mirrl}_{1:m_1 - 1}, \cD_2, \eta} \svert \cD_2, \eta \right]\right] \mper
    \end{align}
    Our aim, now, is to show that 
    \begin{equation}
        \Ex{\exps{\lambda^{\Mirrl}_2 \Paren{\ghate^{\Mirrl}_{m_1} - \vartheta_{m_1}} - \gamma_{E,c}(\lambda^{\Mirrl}_2)  \Paren{\ghate^{\Mirrl}_{m_1} - \hmu_{{m_1}-1}^{\Mirrl}\brackone}^2} \svert \bghate^{\Mirrl}_{1:m_1- 1}, \cD_2, \eta}
    \end{equation}
    is upper bounded by $1$.
    To do so, we define the shorthands $\delta_t \coloneqq \ghate^{\Mirrl}_t - \vartheta_t$ and $\epsilon_t \coloneqq \vartheta_t - \hmu_{t-1}^{\Mirrl}\brackone$, and apply Fan's inequality~\citep{fan2015exponential} which states that for any $\lambda \in [0,1)$ and $\zeta \geq -1$ the following inequality holds: $\exps{\lambda \zeta - \gamma_{E, c}(\lambda) \zeta^2} \leq 1 + \lambda \zeta$. The conditions of Fan's inequality hold in our setting as we can always rescale our random variables or further restrict the the interval from which $\lambda$ takes its values.
    We now proceed to apply Fan's inequality as follows:
    \begin{align}
        &\Ex{\exps{\lambda^{\Mirrl}_2 \Paren{\ghate^{\Mirrl}_{m_1} - \vartheta_{m_1}} - \gamma_{E,c}(\lambda^{\Mirrl}_2)  \Paren{\ghate^{\Mirrl}_{m_1} - \hmu_{{m_1}-1}^{\Mirrl}\brackone}^2} \svert \bghate^{\Mirrl}_{1:m_1- 1}, \cD_2, \eta} \\
        &= \Ex{\exps{\lambda^{\Mirrl}_2 \Paren{\delta_{m_1} - \epsilon_{m_1}} - \gamma_{E,c}(\lambda^{\Mirrl}_2)  \Paren{\delta_{m_1} - \epsilon_{m_1}}^2} \svert \bghate^{\Mirrl}_{1:m_1- 1}, \cD_2, \eta}  \exps{\lambda^{\Mirrl}_2 \epsilon_{m_1}} \\
        &\leq \Ex{\Paren{1 + \lambda^{\Mirrl}_2 \Set{\delta_{m_1} - \epsilon_{m_1}}}  \exps{\lambda^{\Mirrl}_2 \epsilon_{m_1}} \svert \bghate^{\Mirrl}_{1:m_1 - 1}, \cD_2, \eta} \label{eq:proof:cs-fan-fan-application}\\
        &= \Paren{1 + \Ex{\lambda^{\Mirrl}_2 \delta_{m_1} \svert \bghate^{\Mirrl}_{1:m_1 - 1}\cD_2, \eta} - \lambda^{\Mirrl}_2 \epsilon_{m_1}}  \exps{\lambda^{\Mirrl}_2 \epsilon_{m_1}} \\
        &= \Paren{1 - \lambda^{\Mirrl}_2\epsilon_{m_1}}  \exps{\lambda^{\Mirrl}_2 \epsilon_{m_1}} \label{eq:proof:cs-fan-unbiased-estimator} \\
        &\leq \exps{-\lambda^{\Mirrl}_2 \epsilon_{m_1}}  \exps{\lambda^{\Mirrl}_2 \epsilon_{m_1}} =1 \label{eq:proof:cs-fan-exp-ub} \mcom
    \end{align}
    where in \eqref{eq:proof:cs-fan-fan-application} we have used Fan's inequality (given in the proof of \citep[Proposition 4.1]{fan2015exponential}); \eqref{eq:proof:cs-fan-unbiased-estimator} follows because $\ghate^{\Mirrl}_t$ is an unbiased estimator of $\vartheta_t$; and \eqref{eq:proof:cs-fan-exp-ub} follows from the upper-bound $1-x \leq \exps{-x}$.

    The above steps, in turn, imply the following upper-bound
    \begin{align}
        \Ex{E_{m_1}^{\Mirrl}\brackone \svert \cD_2, \eta} &\leq \Ex{\prod_{t=1}^{m_1 - 1} \exps{\lambda^{\Mirrl}_2 \Paren{\ghate^{\Mirrl}_t - \vartheta_t} - \gamma_{E,c}(\lambda^{\Mirrl}_2)  \Paren{\ghate^{\Mirrl}_t - \hmu_{t-1}^{\Mirrl}\brackone}^2} \svert \cD_2, \eta} \\
        &= \Ex{E_{m_1 - 1}^{\Mirrl}\brackone \svert \cD_2, \eta}\mper
    \end{align}
    By recursively applying the above arguments we can conclude that
    \begin{equation}
        \Ex{E_{m_1}^{\Mirrl}\brackone \svert \cD_2, \eta} \leq 1 \mper
    \end{equation}
    That is, $E_{m_1}^{\Mirrl}\brackone$ is a conditional $e$-value, from which it follows that $E_{m_1}^{\Mirrl}\brackone$ is also a marginal $e$-value:
    \begin{equation}
        \Ex{E_{m_1}^{\Mirrl}\brackone} \leq 1 \mper
    \end{equation}

    \paragraphref{Step 2: Applying Markov's inequality.}    
    Applying Markov's inequality to $E_{m_1}$, we have that with probability at least $(1-\alpha)$,
    \begin{equation}
    \lambda^{\Mirrl}_2 \sum_{t=1}^{m_1} \Paren{\ghate^{\Mirrl}_t - \vartheta_t} - \gamma_{E,c}(\lambda^{\Mirrl}_2)  V_{m_1}^{\Mirrl}\brackone < \logs{\flatfrac{1}{\alpha}} \mper
    \end{equation}
    Rearranging the above, we have that with probability at least $(1-\alpha)$,
    \begin{equation}\label{eq:eb-lower-bound-1}
    \sum_{t=1}^{m_1} \vartheta_t > \sum_{t=1}^{m_1} \ghate^{\Mirrl}_t - \frac{\gamma_{E,c}(\lambda^{\Mirrl}_2)  V_{m_1}^{\Mirrl}\brackone + \logs{\flatfrac{1}{\alpha}}}{\lambda^{\Mirrl}_2} \mper
    \end{equation}
    Applying the same reasoning but with $\Dcal_1$ and $\Dcal_2$ swapped, we have that with probability at least $(1-\alpha)$,
    \begin{equation}\label{eq:eb-lower-bound-2}
    \sum_{t=m_1 + 1}^{\widebar T} \vartheta_t > \sum_{t = m_1 + 1}^{\widebar T} \ghate^{\Mirrl}_t - \frac{\gamma_{E,c}(\lambda^{\Mirrl}_1)  V_{m_2}^{\Mirrl}\bracktwo + \logs{\flatfrac{1}{\alpha}}}{\lambda^{\Mirrl}_1}
    \end{equation}
    Union bounding and taking the sum of the left-hand sides of both \eqref{eq:eb-lower-bound-1} and \eqref{eq:eb-lower-bound-2}, we have that with probability at least $(1-2\alpha)$,
    \begin{equation}
    \sum_{t=1}^{\widebar T} \vartheta_t > \sum_{t=1}^{\widebar T} \ghate^{\Mirrl}_t - \frac{\gamma_{E,c}(\lambda^{\Mirrl}_2) V_{m_1}^{\Mirrl}\brackone + \log(1/\alpha)}{\lambda^{\Mirrl}_2} - \frac{\gamma_{E,c}(\lambda^{\Mirrl}_1) V_{m_2}^{\Mirrl}\bracktwo + \log(1/\alpha)}{\lambda^{\Mirrl}_1} \mper
    \end{equation}
    Dividing both sides by $n$ and writing the expression in terms of the individual treatment effect estimators we obtain that with probability $(1-2\alpha)$
    \begin{equation}
      \ate > \frac{1}{n} \sum_{t=1}^T \sum_{i \in g_t} \hite^\Mirrl_\etainvi - \Paren{\frac{\gamma_{E,c}(\lambda^{\Mirrl}_2) V_{m_1}^{\Mirrl}\brackone + \logs{1/\alpha}}{n \lambda^{\Mirrl}_2} + \frac{\gamma_{E,c}(\lambda^{\Mirrl}_1) V_{m_2}^{\Mirrl}\bracktwo + \logs{1/\alpha}}{n \lambda^{\Mirrl}_1}} \eqqcolon L_n^{\EBCI} \mcom
      \label{eq:proof:emp-bern-ci-l}
    \end{equation}
    where $\lambda^{\Mirrl}_1, \lambda^{\Mirrl}_2$ are as given in \eqref{eq:main-paper-emp-Berns-gamma-lambda-definition}. Thus, \eqref{eq:proof:emp-bern-ci-l} forms a lower $(1-2\alpha)$-confidence set.

    \paragraphref{Step 3: Deriving the upper-confidence set.}
    In order to derive the upper confidence set, we repeat steps $1$ and $2$ but with 
    \begin{equation}
        \ghate^{\Mirru}_t = \sum_{i \in g_t} \hite_\etainvi^\Mirru \mper
    \end{equation}
    Putting these steps together we obtain that with probability $(1-2\alpha)$
    \begin{equation}
        \ate < \frac{1}{n}\sum_{t=1}^T \sum_{i \in g_t} \hite_\etainvi^\Mirru + \frac{\gamma_{E,c}(\lambda^{\Mirru}_2) V_{m_1}^{\Mirru}\brackone + \logs{1/\alpha}}{n \lambda^{\Mirru}_2} + \frac{\gamma_{E,c}(\lambda^{\Mirru}_1) V_{m_2}^{\Mirru}\bracktwo + \logs{1/\alpha}}{n \lambda^{\Mirru}_1} \eqqcolon U_n^{\EBCI}\mper
        \label{eq:proof:emp-bern-ci-u}
    \end{equation}
    We, thus, have that $U_n^{\EBCI}$ forms an upper $(1-2\alpha)$-confidence set.
    
    Taking a union bound over the lower- and upper-confidence sets we immediately have that $\Brac{L_n^{\EBCI}, U_n^{\EBCI}}$ forms a $(1-2\alpha)$-confidence interval; thus, completing our proof.
\end{proof}

\begin{proof}[\proofref{corollary:eb-scaling}]
    We will now see why the bound scales as in \eqref{eq:eb-scaling}. Throughout this proof we will be using the following shorthands
    \begin{equation}
        b^*_1 \coloneqq \frac{\gamma_{E,c}(\lambda^{*}_2) V_{m_1}^*\brackone + \log(1/\alpha)}{n\lambda^{*}_2}; \quad b^*_2 \coloneqq \frac{\gamma_{E,c}(\lambda^*_1) V_{m_2}^*\bracktwo + \log(1/\alpha)}{n\lambda^*_1} \mper
    \end{equation}
    for $* \in \Set{\Mirrl, \Mirru}$.
    
    We start by recalling our assumption that $m_1 = m_2 = n/2$. Furthermore, from our assumption that $\hsigma_1^{*2} = \omega(1/n)$ and $\hsigma_2^{*2} = \omega(1/n)$ for $* \in \Set{\Mirrl, \Mirru}$ and two applications of L'H\^opital's rule we can see that $\flatfrac{\gamma_{E,c}(\lambda^{*}_2)}{\Paren{\lambda^{*2}_2 / 2}} \to 1$. 
    We proceed by analyzing $\sqrt{n}  b^*_1$:
    \begin{align}
        \sqrt{n}\frac{\gamma_{E,c}(\lambda^{*}_2) V_{m_1}^*\brackone + \logs{\flatfrac{1}{\alpha}}}{n \lambda^{*}_2} &= \frac{V_{m_1}^*\brackone \Paren{\flatfrac{\lambda^{*2}_2}{2}}  \flatfrac{\gamma_{E,c}(\lambda^{*}_2)}{\Paren{\flatfrac{\lambda^{*2}_2}{2}}} + \logs{\flatfrac{1}{\alpha}}}{\sqrt{n} \lambda^{*}_2} \\
        &\asymp \frac{V_{m_1}^*\brackone \Paren{\flatfrac{\lambda^{*2}_2}{2}} + \logs{\flatfrac{1}{\alpha}}}{\sqrt{n} \lambda^{*}_2} \mper
    \end{align}
    We now recall that by definition $V_{m_1}^{*}\brackone= m_1 \hsigma_1^{*2}$ and that by assumption we have that $m_1 = m_2 = n/2$. Substituting in the value of $\lambda_2^{*}$ and using the previous observations to simplify the terms we have that:
    \begin{align}
        &= \frac{m_1 \hsigma_1^{*2} \Paren{\flatfrac{\logs{1/\alpha}}{m_2 \hsigma_2^{*2}}} + \logs{1/\alpha}}{\sqrt{n}\lambda^{*}_2} \\
        &= \frac{\Set{\Paren{\flatfrac{\hsigma_1^{*2}}{\hsigma_2^{*2}}} + 1} \logs{1/\alpha}}{\sqrt{n} \lambda^{*}_2}\\
        &= \frac{\Set{\Paren{\flatfrac{\hsigma_1^{*2}}{\hsigma_2^{*2}}} + 1} \logs{1/\alpha}}{\sqrt{n} \sqrt{\flatfrac{2 \logs{1/\alpha}}{\Set{(n/2) \hsigma_2^{*2}}}}} \\
        &= \frac{\Set{\Paren{\flatfrac{\hsigma_1^{*2}}{\hsigma^*_2} + \hsigma^*_2}} \logs{1/\alpha}}{2\sqrt{\logs{1/\alpha}}} \\
        &= \Paren{\frac{\hsigma_1^{*2}}{\hsigma_2^*} + \hsigma_2^*}  \frac{\sqrt{\logs{1/\alpha}}}{2} \mper
    \end{align}
    By symmetry we obtain the following characterization on $\sqrt{n} b_2$:
    \begin{equation}
        \sqrt{n}  b^*_2 \asymp \Paren{\frac{\hsigma_2^{*2}}{\hsigma^*_1} + \hsigma^*_1}  \frac{\sqrt{\logs{1/\alpha}}}{2} \mper
    \end{equation}
    Putting the above together and looking at the sum of $b^*_1$ and $b^*_2$ scaled by $\sqrt{n}$ we obtain
    \begin{equation}
        \sqrt{n}\Paren{b^*_1 + b^*_2} \asymp \frac{\sqrt{\logs{1/\alpha}}}{2}   \Paren{\frac{\hsigma_1^{*2}}{\hsigma^*_2} + \frac{\hsigma_2^{*2}}{\hsigma^*_1} + \hsigma^*_1 + \hsigma^*_2}
    \end{equation}
    for $* \in \Set{\Mirrl, \Mirru}$. Thus, completing our proof.
\end{proof}

\section{Proofs of rates of estimation}
\subsection{Upper bound on the $L_2(\PP)$-norm}\label{section:proof-l2p-upper-bound}
In this subsection we present the proofs of \eqref{eq:estimation-upper-bound} which we restate as lemmas.

\begin{lemma}[An estimation error bound from a sub-Gaussian concentration inequality]
\label{lemma:l2p-error-bound}
    Let $\hate$ be the Horvitz-Thompson estimator as in \eqref{eq:mbcr-horvitz-thompson-estimator}, and $\ate$ represent the average treatment effect. Then, under mini-batch complete randomization with $\pi = \flatfrac{1}{K}$ for some integer $K \geq 2$, the $L_2(\PP)$-norm between the Horvitz-Thompson estimator and the average treatment effect can be upper bounded as
    \begin{equation}
        \normtp{\hate - \ate} \leq \frac{2}{\sqrt{n\pi}}\mper
    \end{equation}
\end{lemma}
\begin{proof}[\proofref{lemma:l2p-error-bound}]
    To prove the desired bound, we will rely on the results from \cref{section:proof-hoeffding-mbcr}, and remind the reader that $\ghate_t = \sum_{i \in g_t} \estimator_{\etainvi}$ takes values in $[-G, G]$ where $G = \flatfrac{1}{\pi}$.
    
    We begin the proof by using sub-Gaussian concentration results from \cref{section:proof-hoeffding-mbcr} to bound the tail probability for any $\delta > 0$
    \begin{align}
        \PP\Paren{\frac{1}{n} \sum_{t=1}^T \sum_{i \in g_t} \Paren{\estimator_i - \ite_i} \geq \delta} &= \PP\Paren{\sum_{t=1}^T \sum_{i \in g_t}\Paren{\estimator_i - \ite_i} \geq n \delta} \\
        &= \PP\Paren{\exps{\lambda \sum_{t=1}^T \sum_{i \in g_t}\Paren{\estimator_i - \ite_i}} \geq \exps{\lambda n \delta}} \\
        &= \PP\Paren{\exps{\lambda \sum_{t=1}^T \sum_{i \in g_t} \Paren{\estimator_i - \ite_i} - \frac{TG^2 \lambda^2}{2}} \geq \exps{\lambda n \delta - \frac{TG^2 \lambda^2}{2}}} \\
        &\leq \frac{\Ex{\exps{\sum_{t=1}^T \sum_{i \in g_t} \Paren{\estimator_i - \ite_i} - \frac{TG^2 \lambda^2}{2}}}}{\exps{\lambda n \delta - \frac{TG^2 \lambda^2}{2}}} \label{eq:proof:l2p-markov}\\
        &\leq \frac{1}{\exps{\lambda n \delta - \frac{TG^2 \lambda^2}{2}}} \label{eq:proof:l2p-e-value-bound} \mcom
    \end{align}
    where \eqref{eq:proof:l2p-markov} follows from Markov's inequality and \eqref{eq:proof:l2p-e-value-bound} is due to the exponential random variable being an $e$-value. Letting $\lambda \coloneqq \flatfrac{\delta}{G}$, plugging back in, and simplifying we obtain the following upper bound
    \begin{equation}
        \PP\Paren{\frac{1}{n}\sum_{t=1}^T \sum_{i \in g_t} \Paren{\estimator_i - \ite_i} \geq \delta} \leq \exps{-\frac{\delta^2 n}{2G}} \mper
    \end{equation}
    Taking a union bound we have
    \begin{equation}
        \PP\Paren{\Abs{\frac{1}{n}\sum_{t=1}^T \sum_{i \in g_t} \Paren{\estimator_i - \ite_i}} \geq \delta} \leq 2 \exps{-\frac{\delta^2 n}{2G}} \mper
    \end{equation}
    We now use the expectation-integrated tail identity for nonnegative random variables in order to relate our tail bound to the expectation, giving us
    \begin{align}
        \E_{\PP}\Brac{\Paren{\hate - \ate}^2} &= \int_0^\infty \PP\Paren{\Paren{\hate - \ate}^2 \geq x} dx \\
        &= \int_0^\infty \PP\Paren{\Abs{\hate - \ate} \geq \sqrt{x}} dx \\
        &\leq 2 \int_0^\infty \exps{-\frac{n}{2G}  x} dx\\
        &= \frac{4G}{n} \mper
    \end{align}
    Letting $G = \flatfrac{1}{\pi}$, because of the mini-batch complete randomization procedure, we have that the $L_2(\PP)$-norm $\normtp{\hate - \ate} \equiv \sqrt{\E_{\PP}\Brac{\Paren{\hate - \ate}^2}}$ can be upper bounded as 
    \begin{equation}
        \normtp{\hate - \ate} \leq \frac{2}{\sqrt{n\pi}}\mcom
    \end{equation}
    completing the proof.
\end{proof}

In the following lemma, we provide a proof for deriving an upper bound on $\normtp{\widehat \psi - \psi}$ under Bernoulli randomization that only relies on upper-bounding the variance of $\widehat \psi$ directly. While \cref{lemma:l2p-error-bound} suffices under Bernoulli randomization as well, the following lemma arrives at an upper bound with a smaller constant and it has a simpler proof so we provide it to complement \cref{lemma:l2p-error-bound}.
\begin{lemma}[An upper bound on the RMSE of the Horvitz-Thompson estimator under Bernoulli randomization]
\label{lem:rmse-upper-bound-ht-bernoulli}
    Let $\hate$ be the Horvitz-Thompson estimator from \eqref{eq:ht-estimator} under Bernoulli randomization with some propensity score $\pi \in (0, 1/2]$. We then have that the root mean square error of $\hate$ can be upper bounded as
    \begin{equation}
        \normtp{\hate - \ate} \coloneqq \sqrt{\E_{\PP}\Brac{\Paren{\hate - \ate}^2}} \leq \sqrt{\frac{2}{n\pi}} \mper
    \end{equation}
\end{lemma}

\begin{proof}[Proof of \cref{lem:rmse-upper-bound-ht-bernoulli}]
    We begin by noting that because the Horvitz-Thompson estimator \eqref{eq:ht-estimator} is an unbiased estimator of the average treatment effect under Bernoulli randomization, we have the following series of equalities
    \begin{equation}
        \normtp{\hate - \ate} = \sqrt{\E_{\PP}\Brac{\Paren{\hate - \ate}^2}} = \sqrt{\var{\hate}} \mper
    \end{equation}

    As such, all we need to do is upper bound the variance of the Horvitz-Thompson estimator under Bernoulli randomization. We proceed as follows:
    \begin{align}
        \var{\hate} &= \var{\frac{1}{n} \sum_{i=1}^n \hite_i} \\
        &= \frac{1}{n^2} \sum_{i=1}^n \var{\hite_i}
        \label{eq:proof-rmse-upper-bound-ht-bernoulli-variance} \\
        &= \frac{1}{n^2} \sum_{i=1}^n \Paren{\E_{\PP}\Brac{\hite_i^2} - \E_{\PP}\Brac{\hite_i}^2} \\
        &= \frac{1}{n^2} \sum_{i=1}^n \Paren{\E_{\PP}\Brac{\hite_i^2} - \ite_i^2} \mcom
        \label{eq:proof-rmse-upper-bound-ht-bernoulli-variance-final-decomposition}
    \end{align}
    where \eqref{eq:proof-rmse-upper-bound-ht-bernoulli-variance} follows from our assumption that the potential outcomes are independently drawn (Assumption~\ref{assumption:generalized-potential-outcomes}) and that treatment assignment under Bernoulli randomization is performed independently for each unit. We now focus on the second moment of the individual treatment effect estimators.
    \begin{align}
        \E_{\PP}\Brac{\hite_i^2} &= \E_{\PP}\Brac{\frac{Y_i(1)^2 Z_i}{\pi^2} + \frac{Y_i(0)^2 (1-Z_i)}{(1-\pi)^2}} \\
        &= \frac{1}{\pi^2}\E_{\PP}\Brac{Y_i(1)^2 Z_i} + \frac{1}{(1-\pi)^2}\E_{\PP}\Brac{Y_i(0)^2 (1-Z_i)} \\
       &= \frac{\E_{\PP}\Brac{Y_i(1)^2}}{\pi} + \frac{\E_{\PP}\Brac{Y_i(0)^2}}{1-\pi} \\
       &\leq \frac{1}{\pi} + \frac{1}{1-\pi}
       \label{eq:proof-rmse-upper-bound-ht-bernoulli-potential-outcomes-bound} \\
       &\leq \frac{2}{\pi}
       \label{eq:proof-rmse-upper-bound-ht-bernoulli-one-over-pi-bound} \mcom
    \end{align}
    where \eqref{eq:proof-rmse-upper-bound-ht-bernoulli-potential-outcomes-bound} follows from our assumption that the potential outcomes take values in the unit interval (Assumption~\ref{assumption:generalized-potential-outcomes}), and \eqref{eq:proof-rmse-upper-bound-ht-bernoulli-one-over-pi-bound} follows from the fact that $1/\pi \geq 1/(1-\pi)$ whenever $\pi \in (0,1/2]$ (i.e., the regime we are considering). Hence, combining the above inequalities with \eqref{eq:proof-rmse-upper-bound-ht-bernoulli-variance-final-decomposition} we have that
    \begin{align}
        \var{\hate} &\leq \frac{1}{n^2} \sum_{i=1}^n \Paren{\frac{2}{\pi} - \ite_i^2} \\
        &\leq \frac{1}{n^2} \sum_{i=1}^n \Paren{\frac{2}{\pi}} 
        \label{eq:proof-rmse-upper-bound-ht-bernoulli-ite-squared-positivity-bound}\\
        &= \frac{2}{n\pi} \mcom
    \end{align}
    where \eqref{eq:proof-rmse-upper-bound-ht-bernoulli-ite-squared-positivity-bound} follows from the fact that $-\ite_i^2 \leq 0$. Hence, putting all of the above steps together we have the following upper bound on the root mean square error
    \begin{equation}
        \normtp{\hate - \ate} = \sqrt{\var{\hate}} \leq \sqrt{\frac{2}{n\pi}} \mcom
    \end{equation}
    completing our proof.
    
\end{proof}

\begin{lemma}
    \label{lemma:hoeffding-estimation-bound}
    Let $\hite_i \sim \PP_i$ be a bounded random variable taking values in $[-1, 1]$, and denotes its mean by $\ite_i \coloneqq \Ex{\hite_i}$. Denote the average over $n$ of these, independent, random variables as $\hate \coloneqq \Paren{1/n}\sum_{i=1}^n \hite_i$ and the average of their means as $\ate \coloneqq \Paren{1/n}\sum_{i=1}^n \ite_i$. Then, the $L_2(\PP)$-norm between the sample average and the average of the means is bounded as
    \begin{equation}
        \normtp{\hate - \ate} \leq \frac{2}{\sqrt{n}}\mper
    \end{equation}
\end{lemma}

\begin{proof}[\proofref{lemma:hoeffding-estimation-bound}]
    The proof proceeds in a similar way as that of \cref{lemma:l2p-error-bound}. 
    
    We will first use Hoeffding's inequality to provide an upper bound on the tail probability
    \begin{align}
        \PP\Paren{\frac{1}{n} \sum_{i=1}^n \Paren{\hite_i - \ite_i} \geq \delta} &= \PP\Paren{\sum_{i=1}^n \Paren{\hite_i - \ite_i} \geq n\delta}\\
        &\leq \exps{-\frac{2n^2\delta^2}{4n}} \label{eq:proof-hoeffding-bound}\\
        &= \exps{-\frac{n\delta^2}{2}}\mcom
    \end{align}
    where \eqref{eq:proof-hoeffding-bound} follows from Hoeffding's inequality. Taking a union bound we obtain
    \begin{equation}
        \PP\Paren{\Abs{\frac{1}{n} \sum_{i=1}^n \Paren{\hite_i - \ite_i}} \geq \delta} \leq 2 \exps{-\frac{n\delta^2}{2}}\mper
    \end{equation}
    We now use the expectation-integrated tail identity for nonnegative random variables to relate our tail bound to the expectation, giving us
    \begin{align}
        \E_{\PP}\Brac{\Paren{\hate - \ate}^2} &= \int_{0}^\infty \PP\Paren{\Paren{\hate - \ate}^2 \geq x} dx \\
        &= \int_{0}^\infty \PP\Paren{\Abs{\hate - \ate} \geq \sqrt{x}}dx \\
        &\leq 2 \int_{0}^\infty \exps{-\frac{nx}{2}} dx\\
        &= \frac{4}{n} \mper
    \end{align}
    We, thus, have that the $L_2(\PP)$-norm $\normtp{\hate - \ate} \equiv \sqrt{\E_{\PP}\Brac{\Paren{\hate - \ate}^2}}$ can be upper bounded as
    \begin{equation}
        \normtp{\hate - \ate} \leq \frac{2}{\sqrt{n}} \mcom
    \end{equation}
    completing our proof.
\end{proof}

\subsection{Minimax lower-bound under the stochastic setting}
\label{section:proof-minimax-bound-stochastic}
In this subsection we provide a proof for a more general setting than the one considered in \cref{theorem:minimax-bound-stochastic}. Albeit, the setting of \cref{theorem:minimax-bound-stochastic} can be achieved by letting $d(x, y) = \Abs{x-y}$ and $\Phi(x) = x^2$, in the setting we consider below.

\begin{proof}[\proofref{theorem:minimax-bound-stochastic}]
\label{proof:minimax-bound-stochastic}
    We are tasked with finding a lower bound on the minimax risk given by
    \begin{equation}
        \minimaxrisk \coloneqq \inf_{\estimator} \sup_{\PP \in \cP_\iid} \E_{\PP}\Brac{\Phi\Paren{d(\estimator, \ateP)}} \mcom
    \end{equation}
   where $\cP_\iid$ is the set of $\iid$ distributions as described in \cref{def:iid-stochastic-setting}; $d : \Psi \times \Psi \to \R_{\geq 0}$ is a semi-metric over the parameter space $\Psi$; $\Phi : \R_{\geq 0} \to \R_{\geq 0}$ is an increasing function; and the infimum is taken over all possible estimators.

    Our proof follows a standard Le Cam-style argument but we will write out each step explicitly for the sake of completeness. To this end, let $\delta > 0$ be a constant to be chosen later. Looking at the above expectation for any estimator $\estimator$ and any distribution $\PP \in \cP_\iid$ and applying Markov's inequality,
    \begin{equation}
        \E_{\PP}\Brac{\Phi\Paren{d(\estimator, \ateP)}} \geq \Phi(\delta)  \PP\Paren{\Phi\Paren{d(\estimator, \ateP)} \geq \Phi(\delta)} = \Phi(\delta)  \PP\Paren{d(\estimator, \ateP) \geq \delta}\mcom
    \end{equation}
    where the equality follows from the increasing nature of $\Phi(\cdot)$. Taking the supremum over $\PP \in \cP_\iid$ yields
    \begin{equation}
        \sup_{\PP \in \cP_\iid} \E_{\PP}\Brac{\Phi\Paren{d(\estimator, \ateP)}} \geq \Phi(\delta) \sup_{\PP \in \cP_\iid}  \PP\Paren{d(\estimator, \ateP) \geq \delta}\mper
    \end{equation}
   It, thus, suffices to provide a lower bound for $\sup_{\PP \in \cP_\iid} \PP\Paren{d(\estimator, \ateP) \geq \delta}$.  

   We will let $\Psi^\brackM \coloneqq \Set{\atePo,\dots,\atePM}$ be a set of $2\delta$-separated parameters contained in the space $\Psi$. That is, for all $\atePi, \atePj \in \Psi^\brackM$ such that $i \neq j$ we have that $d(\atePi, \atePj) \geq 2\delta$.
   We will associate a representative distribution $\PP_i$ to each parameter $\atePi \in \Psi^\brackM$, and we will denote the collection of these distributions as $\cP_\iid^\brackM$. Notice how an average always lower bounds a supremum. So, taking an average over the distributions in $\cP_\iid^\brackM$ gives us
   \begin{equation}
       \Phi(\delta) \sup_{\PP \in \cP_\iid} \PP\Paren{d(\estimator, \ateP) \geq \delta} \geq \Phi(\delta)  \frac{1}{M} \sum_{i=1}^M \PP_i\Paren{d(\estimator, \ateP) \geq \delta} \mper
    \label{eq:proof:lecam-average-lower-bound}
   \end{equation}
   Now, looking to lower-bound the average of probabilities in \eqref{eq:proof:lecam-average-lower-bound} by the average of testing errors, we define the test $\theta$ given by
   \begin{equation}
       \theta \coloneqq \argmin_{k \in [M]} d(\estimator, \atePk) \mper
   \end{equation}
   By construction of $\theta$, if $d(\estimator, \atePk) < \delta$, then it must be the case that $\theta = k$ by virtue that all $\atePj \in \Psi^\brackM$ are $2\delta$-separated. In other words, $\Set{d(\estimator, \atePk) < \delta} \subseteq \Set{\theta = k}$, and hence the reverse inclusion holds for their complements
   \begin{equation}
       \Set{d(\estimator, \atePk) \geq \delta} \supseteq \Set{\theta \neq k} \mcom
   \end{equation}
   implying the following inequality on the $\PP_k$-probabilities of the above events
   \begin{equation}
       \PP_k\Paren{d(\estimator, \atePk) \geq \delta} \geq \PP_k(\theta \neq k)
   \end{equation}
   for each $k \in [M]$. Plugging the above lower bound into \eqref{eq:proof:lecam-average-lower-bound} we obtain
   \begin{equation}
       \sup_{\PP \in \cP_\iid} \E_{\PP}\Brac{\Phi\Paren{d(\estimator, \ateP)}} \geq \Phi(\delta)  \frac{1}{M} \sum_{j=1}^M \PP_j\Paren{\theta \neq j} \mper
       \label{eq:proof:lecam-test-lower-bound}
   \end{equation}
   We now take the infimum over all possible estimators on the left hand side of \eqref{eq:proof:lecam-test-lower-bound}, and the infimum over the set of all tests induced by the estimators on the right hand side. Notice how the infimum over all possible tests is by construction smaller than the infimum over all possible estimators, giving us
   \begin{equation}
        \inf_{\estimator} \sup_{\PP \in \cP_\iid} \E_{\PP}\Brac{\Phi\Paren{d(\estimator, \ateP)}} \geq \Phi(\delta)  \inf_{\theta} \frac{1}{M} \sum_{j=1}^M \PP_j\Paren{\theta \neq j}
    \label{eq:proof:lecam-lower-bound}
   \end{equation}

   Now that we have lower bounded the minimax estimation risk by the average testing error over $M$ distributions, we are ready to follow Le Cam's two point method. To do so, we will let $\cP_\iid^\brackM \coloneqq \Set{\PP_0, \PP_1}$, implying that the lower bound in \eqref{eq:proof:lecam-lower-bound} can be characterized as
   \begin{equation}
       \Phi(\delta)  \inf_{\theta} \frac{1}{2} \Set{\PP_0\Paren{\theta \neq 0} + \PP_1\Paren{\theta \neq 1}}\mper
   \end{equation}
   Notice how in the above equation we are considering the Bayes risk. So, relating the Bayes risk to the total variation (TV) distance we have that
   \begin{equation}
       \Phi(\delta)  \inf_{\theta} \frac{1}{2} \Set{\PP_0\Paren{\theta \neq 0} + \PP_1\Paren{\theta \neq 1}} = \Phi(\delta)  \frac{1}{2} \Set{1 - \Norm{\PP_1 - \PP_0}_{\tv}}\mcom
   \end{equation}
   noting that the right-hand side no longer depends on any particular estimator $\estimator$ or test $\theta$. Putting all of the above together and applying Pinsker's inequality to upper-bound the total variation distance in terms of the Kullback-Leibler (KL) divergence, we have
   \begin{equation}
        \inf_{\estimator} \sup_{\PP \in \cP_\iid} \E_{\PP}\Brac{\Phi\Paren{d(\estimator, \ateP}} \geq \frac{\Phi(\delta)}{2}  \Set{1 - \Norm{\PP_1 - \PP_0}_{\tv}}
    \geq \frac{\Phi(\delta)}{2}  \Set{1 - \sqrt{\frac{1}{2}D_{\KL}\Paren{\PP_0 \| \PP_1}}}\mper
   \end{equation}
   
   In order to continue with the two-point Le Cam argument, we now proceed to characterize the distributions $\PP_1$ and $\PP_0$ so that the estimands $\atePz$ and $\atePo$ are $2\delta$-separated.
   Indeed, let $\PP_0$ and $\PP_1$ be the probability distributions such that $Y_i(0) \simiid \mathrm{Bernoulli}(1/2)$ and $Y_i(0) \independent Y_i(1)$ under both $\PP_0$ and $\PP_1$ and so that
    \begin{align}
        \{ Y_i(1) \}_{i=1}^n &\simiid \mathrm{Bernoulli}\left ( \frac{1-2\delta}{2} \right ) \quad \text{under $\PP_0$, and}\\
        \{ Y_i(1) \}_{i=1}^n &\simiid \mathrm{Bernoulli}\left ( \frac{1+2\delta}{2} \right ) \quad \text{under $\PP_1$},
    \end{align}
    noting that $\widebar \psi$ is $2\delta$-separated under the two distributions. So, it remains to inspect $D_{\KL}\Paren{\PP_0 \| \PP_1}$.
    
    Indeed, exploiting the fact that the vector $\bZ \coloneqq (Z_1, \dots, Z_n)$ is independent of the potential outcomes $\{ (Y_i(0), Y_i(1) \}_{i=1}^n$ which are themselves $\iid$, we will proceed to upper bound $D_{\KL}(\PP_0 \| \PP_1)$.
    
    Notice that conditional on $\bZ$, the vector $\bY \coloneqq \Paren{Y_1, \dots, Y_n}$ consists of exactly $n_0$ $\iid$ draws from the marginal distribution of $Y(0)$ and $n_1$ $\iid$ draws from the marginal distribution of $Y(1)$ and hence the Kullback-Leibler divergence can be written as
    
    \begin{align}
        D_{\KL}\Paren{\PP_0 \| \PP_1} &= \E_{\PP_0}\Brac{\logs{\frac{d\PP_0\left (\bZ, \bY\right )}{d\PP_1\left (\bZ, \bY\right )}}} \\
        &= \E\Brac{\E_{\PP_0}\Brac{\logs{\frac{d\PP_0\left (\bY \svert \bZ \right )  \cancel{d\PP_0\Paren{\bZ}}}{d\PP_1\left (\bY \svert \bZ\right )  \cancel{d\PP_1\Paren{\bZ}}}} \svert \bZ}} \label{eq:proof:minimax-kl-decomp-marginal-z}\\
        &= \E\Brac{\E_{\PP_0}\Brac{\logs{\frac{d\PP_0\Paren{\bY(1) \svert \bZ}}{d\PP_1\Paren{\bY(1) \svert \bZ}}} + \logs{\frac{\cancel{d\PP_0\Paren{\bY(0) \svert \bZ}}}{\cancel{d\PP_1\Paren{\bY(0) \svert \bZ}}}} \svert \bZ}}\label{eq:proof:minimax-kl-decomp-marginal-y0}\\
        &= \E\Brac{n_1 \E_{\PP_0}\Brac{\logs{\frac{d\PP_0\Paren{Y(1) \svert \bZ}}{d\PP_1\Paren{Y(1) \svert \bZ}}} \svert \bZ}}\label{eq:proof:minimax-kl-decomp-tensorization}\\ 
        &= n\pi  D_{\KL}\Paren{\PP_0\Paren{Y(1)} \| \PP_1\Paren{Y(1)}}, \label{eq:proof:minimax-kl-decomp-independence-from-z}
    \end{align}
    where \eqref{eq:proof:minimax-kl-decomp-marginal-z} follows because the treatment assignment vectors have the same marginal distribution under $\PP_0$ and $\PP_1$; \eqref{eq:proof:minimax-kl-decomp-marginal-y0} uses the fact that the potential outcomes under control are $\iid$ draws from the same marginal distribution; \eqref{eq:proof:minimax-kl-decomp-tensorization} follows from the tensorization property of the Kullback-Leibler divergence; and \eqref{eq:proof:minimax-kl-decomp-independence-from-z} follows from our ignorability assumption (Assumption~\ref{assumption:ignorability}) and our assumption that $\Ex{n_1} = n\pi$.
    Now, continuing from \eqref{eq:proof:minimax-kl-decomp-independence-from-z}, we can upper-bound the KL divergence for any $\delta \in (0, 1/4)$ as
    \begin{align}
        D_{\KL} \Paren{\PP_0 \| \PP_1}  &= n\pi  D_{\KL}\Paren{\PP_0\Paren{Y(1)} \| \PP_1\Paren{Y(1)}} \\
        &= n\pi  \Paren{\flatfrac{(1-2\delta)}{2}}  \logs{\frac{\flatfrac{(1-2\delta)}{2}}{\flatfrac{(1+2\delta)}{2}}} + \Paren{\flatfrac{(1+2\delta)}{2}}  \logs{\frac{\flatfrac{(1+2\delta)}{2}}{\flatfrac{(1-2\delta)}{2}}}\\
        &= n\pi  2\delta  \logs{\flatfrac{(1 + 2\delta)}{(1-2\delta)}} \\
        &\leq n\pi 9\delta^2 \mper
    \end{align}
    Putting all of the previous steps together and returning to the minimax risk $\minimaxrisk$ combined with our implied upper-bound on the total variation distance, we have for any $\delta \in (0, 1/4)$,
    \begin{equation}
        \minimaxrisk \equiv \inf_{\estimator} \sup_{\PP \in \cP_\iid} \E_{\PP}\Brac{\Phi\Paren{d(\estimator, \ateP)}} \geq \frac{\Phi(\delta)}{2}  \Set{1 - \sqrt{\frac{1}{2} n_1  9\delta^2}} \mper
    \end{equation}
    Setting $\delta^2 = \flatfrac{1}{(16 n\pi)}$ we obtain the following lower bound
    \begin{equation}
        \minimaxrisk \geq \frac{\Phi\Paren{\flatfrac{1}{\sqrt{16 n\pi}}}}{2}  \Set{1 - \frac{3}{4 \sqrt{2}}} \mcom
    \end{equation}
    and if we let $\Phi(t) = t^2$ the above can be simplified to
    \begin{equation}
        \minimaxrisk \geq \frac{1}{n\pi} \Paren{\frac{1}{32} - \frac{3}{128 \sqrt{2}}} \geq \frac{1}{72 n\pi}\mcom
    \end{equation}
    which completes the proof of \cref{theorem:minimax-bound-stochastic}.
\end{proof}

\subsection{Minimax lower bound under the design-based setting}
\label{section:proof-minimax-lower-bound-design-based}
\begin{proof}[\proofref{lemma:minimax-relation-stochastic-design-based}]
    We are now interested in relating the minimax risk of the stochastic setting to the minimax risk under the design-based setting
    \begin{equation}
    \label{eq:proof:minimax-design-based}
        \minimaxrisk \coloneqq \inf_{\estimator}~\sup_{\PP \in \cP_{\DB}} \E_{\PP}\Brac{\Phi\Paren{d(\estimator, \sate)}} \mcom
    \end{equation}
    where $\cP_\DB$ is the set of distributions as described in \cref{def:design-based-setting}. In particular, notice that for any $\PP \in \cP_\DB$ the estimand $\sate$ is no longer $\EE_{\PP}[Y(1) - Y(0)]$ as it was in the proof of \cref{theorem:minimax-bound-stochastic}, but is instead now given by
    \begin{equation}
        \sate = \frac{1}{n}\sum_{i=1}^n (y_i(1) - y_i(0)),
    \end{equation}
    where the particular values of $(y_i(0),y_i(1))_{i=1}^n$ are entirely determined by $\PP \in \cP_\DB$

    The proof will proceed by reducing the problem of lower bounding \eqref{eq:proof:minimax-design-based} to that of lower bounding the minimax risk over the collection of distributions $\cP_\iid$ and showing that the difference thereof is negligible due to $\sate$ being a ``good'' estimator of the superpopulation mean $\ateP$ in a sense that will be made formal shortly.
    For this reason, the proof will contain discussions of two collections of distributions $\cP_\iid$ and $\cP_\DB$ defined in \cref{def:iid-stochastic-setting,def:design-based-setting}, respectively. Moreover, we will also be referring to two average treatment effects: (1) the \emph{sample} average treatment effect; (2) the \emph{population} average treatment effect. We will denote the sample average treatment effect as $\sate \coloneqq (\flatfrac{1}{n}) \sum_{i=1}^n (Y_i(1) - Y_i(0))$ and the population average treatment effect as $\ateP \coloneqq \E_{\PP} \Brac{Y(1) - Y(0)}$.

    Having established the necessary notation and concepts, we now proceed to upper bound the risk of any estimator $\estimator$ over the collection of $\iid$ distributions $\cP_\iid$ via the triangle inequality
    \begin{equation}
        \sup_{\PP \in \cP_\iid} \E_{\PP}\Brac{\Phi\Paren{d(\estimator, \ateP)}} \leq \sup_{\PP \in \cP_\iid} \E_{\PP}\Brac{\Phi\Paren{d(\estimator, \sate)}} + \underbrace{\sup_{\PP \in \cP_\iid} \E_{\PP}\Brac{\Phi\Paren{d(\sate, \ateP)}}}_{(\star)}\mper
    \end{equation}
    Notice how $(\star)$ does not depend on the estimator $\estimator$ at all, and hence we have after taking infima over estimators,
    \begin{equation}
        \inf_{\estimator}~\sup_{\PP \in \cP_\iid} \E_{\PP}\Brac{\Phi\Paren{d(\estimator, \ateP)}} \leq \inf_{\estimator}~\underbrace{\sup_{\PP \in \cP_\iid} \E_{\PP}\Brac{\Phi\Paren{d(\estimator, \sate)}}}_{(\dagger)} + \sup_{\PP \in \cP_\iid} \E_{\PP}\Brac{\Phi\Paren{d(\sate, \ateP)}}\mper
    \end{equation}
    Furthermore, notice that we can upper-bound $(\dagger)$ in terms of a supremum over $\cP_\DB$ as follows. Applying \cref{lemma:cond-expectation-bounds-expectation} with $X = \Phi\Paren{d(\estimator, \sate)}$ and $\cA$ being the set of all possible  $\{0, 1\}^n \times \{0, 1\}^n$-valued potential outcome tuples $(Y_i(0), Y_i(1))_{i=1}^n$, we have that for any $\PP \in \cP_\iid$
    \begin{equation}
        \E_{\PP}\Brac{\Phi\Paren{d(\estimator, \sate)}} \leq \sup_{\mathbb{Q} \in \cP_\DB} \E_{\mathbb{Q}}\Brac{\Phi\Paren{d(\estimator, \sate}}
    \end{equation}
    and hence the same inequality holds when taking a supremum over $\PP \in \cP_\iid$.
    Thus, putting all of the above steps together we obtain the following upper bound on the minimax risk under the stochastic setting
    \begin{align}
        &\inf_{\estimator}~\sup_{\PP \in \cP_\DB} \E_{\PP}\Brac{\Phi\Paren{d(\estimator, \sate)}} + \sup_{\PP \in \cP_\iid} \E_{\PP}\Brac{\Phi\Paren{d(\sate, \ateP)}}\\
        &\qquad \geq \inf_{\estimator}~\sup_{\PP \in \cP_\iid} \E_{\PP}\Brac{\Phi\Paren{d(\estimator, \ateP)}}
        \label{eq:proof:minimax-risk-stochastic-bound}
    \end{align}
    We can equivalently write \eqref{eq:proof:minimax-risk-stochastic-bound} as
    \begin{align}
        &\inf_{\estimator}~\sup_{\PP \in \cP_\DB} \E_{\PP}\Brac{\Phi\Paren{d(\estimator, \sate)}} \\
        &\quad \geq \inf_{\estimator}~\sup_{\PP \in \cP_\iid} \E_{\PP}\Brac{\Phi\Paren{d(\estimator, \ateP)}} - \sup_{\PP \in \cP_\iid} \E_{\PP}\Brac{\Phi\Paren{d(\sate, \ateP)}} \mcom
    \label{eq:proof:minimax-risk-stochastic-bound-rearranged}
    \end{align}
    which is what we wanted to show; thus, completing our proof.
\end{proof}

\begin{proof}[\proofref{theorem:minimax-lower-bound-design-based}]
    The proof of the theorem consists of lower bounding the right hand side of \eqref{eq:proof:minimax-risk-stochastic-bound-rearranged}.
    From \cref{theorem:minimax-bound-stochastic} we have the following lower bound on the minimax risk under the stochastic setting
    \begin{equation} 
        \fM\Paren{\estimator \Paren{\cP_{\iid}}} \geq \frac{\Phi\Paren{\flatfrac{1}{\sqrt{16 n \pi}}}}{2}  \Set{1 - \frac{3}{4 \sqrt{2}}} \mcom
    \end{equation}
    which implies the following general lower-bound on \eqref{eq:proof:minimax-risk-stochastic-bound}
    \begin{equation}
        \inf_{\estimator}~\sup_{\PP \in \cP_\DB} \E_{\PP}\Brac{\Phi\Paren{d(\estimator, \sate)}} \geq \frac{\Phi\Paren{\flatfrac{1}{\sqrt{16 n\pi}}}}{2}  \Set{1 - \frac{3}{4 \sqrt{2}}} - \sup_{\PP \in \cP_\iid} \underbrace{\E_{\PP}\Brac{\Phi\Paren{d(\sate, \ateP)}}}_{(\star)}\mper
    \end{equation}
    The above inequality is general in nature, as we have not instantiated the semi-metric $d(\cdot,\cdot)$ nor the increasing function $\Phi(\cdot)$. Moreover, notice how the second term in the inequality consists of an ``oracle'' estimator term of the average treatment effect. In other words, it consists of the average of the difference of the potential outcomes of $n$ units.

    We now instantiate the problem, letting $d(x, y) = \Abs{x - y}$ and $\Phi(x) = x^2$, so that the quantity we are controlling is the squared $L_2(\PP)$-norm. From \cref{lemma:hoeffding-estimation-bound} we have the following upper bound on $(\star)$
    \begin{equation}
        \normtp{\sate - \ateP} \leq \frac{2}{\sqrt{n}}
    \end{equation}
    and from \cref{theorem:minimax-bound-stochastic} we have the following lower bound on the minimax risk under the stochastic setting
    \begin{equation}
        \inf_{\estimator}~\sup_{\PP \in \cP_\iid} \E_{\PP}\Brac{\Paren{d(\estimator, \ateP)}^2} \geq \frac{1}{72 n\pi}\mper
    \end{equation}
    Putting both of these bounds together yields the desired lower bound on the minimax risk under the design-based setting
    \begin{equation}
        \inf_{\estimator} \sup_{\PP \in \cP_\DB} \normtp{\estimator - \sate} \geq \frac{1}{\sqrt{72 n\pi}} - \frac{2}{\sqrt{n}} \mper
    \end{equation}
    Thus, completing the proof of \cref{theorem:minimax-lower-bound-design-based}.

\end{proof}

\section{Other and auxiliary results}
\subsection{Equivalence between mini-batch complete randomization and complete randomization}\label{section:equivalence-mbcr-cr}
\begin{proof}[\proofref{proposition:mbcr-cr-equivalence}]
To show the equivalence between mini-batch complete randomization and complete randomization, we need to argue that the distribution they induce over all the possible assignments satisfying the conditions that there are $n_1$ ones and $n_0$ zeros are the same. First, recall that the probability of any assignment $\bz$ under complete randomization is $\PP_{\CR}\Paren{\bZ = \bz} = \flatfrac{1}{{n \choose n_1}}$. Hence, all we need to do is argue that mini-batch complete randomization produces the same assignment probability.

To show the equivalence, notice how mini-batch complete randomization consists of sampling from two, independent, uniform distributions over two different sets of all possible permutations, i.e., we first sample $\beta$ and then sample $\eta$. Now, if we condition on $\beta$, and apply the permutation $\eta$ to $a$, we can see that 
\begin{align}
    \PP_{\MBCR}\Paren{\bZ = \bz \svert \beta} = \flatfrac{1}{{n \choose n_1}} \mper
\end{align}
This is because $\bZ \equiv \Paren{Z_1, \dots, Z_n} = \Paren{a_{\rho(1)}, \dots, a_{\rho(n)}}$. That is, $\bZ$ is the result from sampling uniformly at random from the set that contains vectors with $n_1$ ones and $n_0$ zeros. As a last step, notice how since $\PP_{\MBCR}\Paren{\bZ = \bz \svert \beta} = \flatfrac{1}{{n \choose n_1}}$ is a constant, marginalizing over $\beta$ does not change this quantity. Hence, we have shown that mini-batch complete randomization is equivalent to complete randomization, completing our proof.
\end{proof}

\subsection{Unbiasedness of the Horvitz-Thompson estimator}\label{section:unbiased-ht}
\begin{proof}[\proofref{proposition:unbiased-ht}]
    Our strategy is to show that 
    \begin{equation}
        \hite_\etainvi \coloneqq Y_{\eta^{-1}(i)} \Paren{\frac{Z_{\beta(i)}}{\flatfrac{1}{G}} - \frac{\Paren{1-Z_{\beta(i)}}}{1 - \flatfrac{1}{G}}}
    \end{equation}
    is a conditionally, on $\eta$, unbiased estimator of the individual treatment effect, \\$\ite_{\etainvi} \equiv \Ex{Y_{\etainvi}(1) - Y_{\etainvi}(0) \svert \eta}$. From this, it then follows that $\hate$ is an unbiased estimator of the average treatment effect.

    We proceed by taking the conditional expectation
    \begin{align}
        \Ex{\hite_\etainvi \svert \eta} &= \Ex{Y_{\etainvi} \Paren{\frac{Z_{\beta(i)}}{1/G} - \frac{(1-Z_{\beta(i)})}{1 - 1/G}} \svert \eta} \\
        &= \Ex{\frac{Y_\etainvi (1) Z_{\beta(i)}}{1/G} - \frac{Y_\etainvi(0) (1-Z_{\beta(i)})}{1 - 1/G} \svert \eta} \label{eq:proof-ate-sutva} \\
        &= \frac{1}{1/G}\Ex{Y_{\etainvi}(1) Z_{\beta(i)} \svert \eta} - \frac{1}{1 - 1/G} \Ex{Y_{\etainvi}(0) (1 - Z_{\beta(i)}) \svert \eta} \\
        &= \frac{1}{1/G} \Ex{Y_{\etainvi}(1) \svert \eta} \Ex{Z_{\beta(i)}} - \frac{1}{1 - 1/G} \Ex{Y_{\etainvi}(0) \svert \eta} \Ex{1 - Z_{\beta(i)}} \label{eq:proof-ate-ignorability} \\
        &= \frac{1/G}{1/G} \Ex{Y_{\etainvi}(1) \svert \eta} - \frac{1 - 1/G}{1 - 1/G} \Ex{Y_{\etainvi}(0) \svert \eta} \\
        &= \Ex{Y_{\etainvi}(1) \svert \eta} - \Ex{Y_{\etainvi}(0) \svert \eta} \mcom
    \end{align}
    where \eqref{eq:proof-ate-sutva} follows from the stable unit treatment value (SUTVA) assumption, and \eqref{eq:proof-ate-ignorability} follows from the ignorability assumption and the independence between $\beta$ and $\eta$. We have, thus, shown that $\Ex{\hite_{\etainvi} \svert \eta} = \ite_{\etainvi}$, and a similar result follows for the estimators in the group $\widebar g$. 
    
    Due to the linearity of expectations, we can show that the Horvitz-Thompson estimator is indeed a, conditionally, unbiased estimator of the average treatment effect,
    \begin{align}
        \Ex{\hate \svert \eta} &\equiv \Ex{\frac{1}{n} \sum_{t = 1}^T \sum_{i \in g_t} \hite_{\etainvi} + \frac{1}{n}\sum_{i \in \widebar g} \hite_{\etainvi} \svert \eta} \\
        &= \frac{1}{n} \Paren{\sum_{t = 1}^T \sum_{i \in g_t} \Ex{\hite_{\etainvi} \svert \eta} + \sum_{i \in \widebar g} \Ex{\hite_{\etainvi} \svert \eta}}\\
        &= \frac{1}{n} \Paren{\sum_{t=1}^T \sum_{i \in g_t} \ite_{\etainvi} + \sum_{i \in \widebar g} \ite_{\etainvi}} \\
        &= \frac{1}{n} \sum_{i=1}^n \ite_i \\
        &= \ate
    \end{align}
    By the law of total expectation we have that $\Ex{\hate} = \Ex{\Ex{\hate \svert \eta}} = \ate$; thus, showing that the Horivtz-Thompson estimator is a conditionally and marginally unbiased estimator of the average treatment effect.
\end{proof}

\subsection{Conditional Expectation bounds the Expectation}
\begin{lemma}
\label{lemma:cond-expectation-bounds-expectation}
    Let $X \sim \PP$ be a random variable and $\cA$ be the set of exclusive and exhaustive events $A$. Then, we have that we can bound the expectation of $X$ by the supremum of $\Ex{X \svert A}$ over all events in $\cA$. Formally, that is
    \begin{equation}
        \Ex{X} \leq \sup_{A \in \cA} \Ex{X \svert A} \mper
    \end{equation}
\end{lemma}
\begin{proof}[\proofref{lemma:cond-expectation-bounds-expectation}]
    Let $X \sim \PP$ be a random variable and $\cA$ be the set of exclusive and exhaustive events $A$. Notice how by the law of total expectation we have the following identity
    \begin{equation}
        \Ex{X} = \sum_{A \in \cA} \Ex{X \svert A}  \PP(A) \label{eq:proof-total-exp}\mper
    \end{equation}
    We can then take the supremum of $\EE[X \mid A]$ over all events in $\cA$ giving us the following upper bound on \eqref{eq:proof-total-exp}
    \begin{align}
        \sum_{A \in \cA} \Ex{X \svert A}  \PP(A) &\leq \sum_{A \in \cA} \sup_{B \in \cA} \Ex{X \svert B}  \PP(A) = \sup_{B \in \cA} \Ex{X \svert B} \mper
    \end{align}
    Thus, completing our proof. 
\end{proof}

\subsection{Bounds on the group-wise Horvitz-Thompson estimator}
\begin{lemma}[Bounds on the group-wise estimator.]
\label{lemma:group-wise-estimator-bound}
    Under Assumption~\ref{assumption:generalized-potential-outcomes}, the group-wise Horvitz-Thompson estimator, $\ghate$, under mini-batch complete randomization can be bounded as
    \begin{equation}
        -G \leq \sum_{i \in g_t} Y_{\etainvi} \Paren{\frac{Z_\betai}{\flatfrac{1}{G}} - \frac{1 - Z_{\betai}}{1 - \flatfrac{1}{G}}} \leq G \mper
    \end{equation}
\end{lemma}
\begin{proof}[\proofref{lemma:group-wise-estimator-bound}]
   Recall that in each group $g_t$ only one unit is assigned to the treatment while the rest, $G-1$ units, are assigned to the control. Moreover, recall that the potential outcomes take values within the unit interval $y_i(q) \in [0,1]$ for $q \in {0,1}$. We can, thus, see that $\ghate_t$ can be upper bounded as
   \begin{equation}
       \ghate_t \leq \frac{1}{\flatfrac{1}{G}} = G \mcom
   \end{equation}
    which is achieved when the outcome of the unit under treatment takes on the value $1$, while the outcomes of the units under control are all equal to $0$. The group-wise estimator is lower bounded as
    \begin{equation}
        \ghate_t \geq \frac{1}{1- \flatfrac{1}{G}}  -(G-1) = -G  \frac{G-1}{G-1} \mcom
    \end{equation}
    which is achieved when the outcome of the unit under treatment takes $0$ as its value, and the outcomes of units under control are all $1$s. Thus, completing the proof.
\end{proof}

\end{document}